%% file: 0_MAIN.tex
\documentclass{article}

\pdfoutput=1
\usepackage{arxiv}

\usepackage[utf8]{inputenc} % allow utf-8 input
\usepackage[T1]{fontenc} %\usepackage{ae} \usepackage{aecompl} % use 8-bit T1 fonts
\usepackage{hyperref}       % hyperlinks
\usepackage{url}            % simple URL typesetting
\usepackage{booktabs}       % professional-quality tables
\usepackage{amsfonts}       % blackboard math symbols
\usepackage{nicefrac}       % compact symbols for 1/2, etc.
\usepackage{microtype}      % microtypography
\usepackage{lipsum}

\usepackage{bbm}           % allows to write the "blackboard 1" symbol \mathbbm{1}
\usepackage{amsthm}
\usepackage{amsmath,amssymb}
\usepackage[]{algorithm2e}

\usepackage{xfrac}
\usepackage{pgf,tikz,pgfplots}
\pgfplotsset{compat=1.14}
\usepackage{mathrsfs}
\usetikzlibrary{arrows}

\newcommand{\R}{\mathbb{R}_{\geq 0}}
\newcommand{\eps}{\varepsilon}
\DeclareMathOperator*{\argmax}{arg\,max}

\DeclareMathOperator*{\st}{stop}
\DeclareMathOperator{\ins}{Insert}

\DeclareMathOperator{\maxprescore}{MaxPscore}
\DeclareMathOperator{\maxscore}{MaxScore}
\DeclareMathOperator{\interval}{FindInterval}
\DeclareMathOperator{\MMS}{MMS}
\DeclareMathOperator{\lazy}{LazyMMS}
\DeclareMathOperator{\phragmen}{seqPhragmen}
\DeclareMathOperator{\maxphragmen}{maxPhragmen}
\DeclareMathOperator{\phragmms}{Phragmms}
\DeclareMathOperator{\LSPJR}{LS-Phragmms}
\DeclareMathOperator{\local}{LS-Phragmms}
\DeclareMathOperator{\bal}{Bal}
\DeclareMathOperator{\supp}{supp}
\DeclareMathOperator{\score}{score}

\DeclareMathOperator{\prescore}{pscore}
\DeclareMathOperator{\slack}{slack}

\newtheorem{theorem}{Theorem}

\newtheorem{lemma}[theorem]{Lemma}
\newtheorem*{lemma*}{Lemma}
\newtheorem{definition}[theorem]{Definition}
\newtheorem{remark}[theorem]{Remark}
\newtheorem*{heuristic}{Heuristic}

\title{A verifiably secure and proportional committee election rule}
\author{
  Alfonso Cevallos \\
  Web 3.0 Technologies Foundation\\
  Zug, Switzerland \\
  \texttt{alfonso@web3.foundation} \\
   \And
 Alistair Stewart \\
  Web 3.0 Technologies Foundation\\
  Zug, Switzerland \\
  \texttt{alistair@web3.foundation} \\
}

\begin{document}
\maketitle

\input{abstract.tex}


%\newpage

\input{intro.tex}

\input{preliminaries.tex}

\input{complexity.tex}

\input{heuristic.tex}

\input{approx315.tex}

\input{verifiable.tex}

\input{local.tex}

\input{implementation.tex}

\input{conclusion.tex}

\bibliographystyle{abbrv}  
\bibliography{references} 

\appendix

\input{balanced.tex}

\input{flow.tex}

\input{algorithms}

\input{lazymms.tex}

\input{proofs.tex}

\end{document}

%% file: abstract.tex
\begin{abstract}

The concept of proportional representation in approval-based committee elections has appeared in the social choice literature for over a century and is typically understood as avoiding the underrepresentation of minorities. 
However, we argue that the security of some distributed systems critically depends on the opposite goal of preventing the \emph{overrepresentation} of any minority, a goal not previously formalized that leads us to an optimization objective known as \emph{maximin support}. 
After providing a thorough analysis of the computational complexity of this objective, we propose a new efficient election rule 
%inspired in Phragm\'{e}n's methods 
that simultaneously achieves a) a constant-factor approximation guarantee for it, and b) the property of \emph{proportional justified representation} (PJR) -- one of the strongest forms of proportional representation. However, the most striking feature of the new rule is that one can \emph{verify in linear time} that the winning committee satisfies the two aforementioned guarantees, even when the algorithm is executed by an untrusted party who only communicates the output. 
As a result, the rule can be adapted into a \emph{verifiable computing scheme}. 
Moreover, its verification procedure easily admits \emph{parallel processing} for further efficiency.
%Finally, we present an efficient post-computation that, when paired with any approximation algorithm for maximin support, returns a new solution that a) preserves the approximation guarantee and b) can be efficiently verified to satisfy PJR.

Our work is motivated by an application on blockchain networks that implement \emph{Nominated Proof-of-Stake}, where the community elects a committee of validators to participate in the consensus protocol, and where preventing overrepresentation protects the network against attacks by an adversarial minority. 
Our election rule enables a validator selection protocol with formal guarantees on security and proportionality, and its adaptation as a verifiable computing scheme with a parallelized verification proves to be key for its successful implementation given the computationally limited nature of the blockchain architecture. 
%We provide details of such an implementation in the \emph{Polkadot} network, launched in 2020.

\end{abstract}

%% file: intro.tex
\section{Introduction}

In an approval-based committee election, a voter either approves or disapproves of each candidate, with no limit on the number of approved ones and no declared preferences among them~\cite{lackner2020approval}. From the voters' ballots taken as input, an election rule outputs a winning committee of candidates of a certain size $k$, in the pursuit of some goals or criteria. Proportional representation is one of the most prominent such criteria in the literature of these election rules. 
It is typically understood as a guarantee that small minorities within the electorate are not \emph{underrepresented} by the winning committee, and is considered an imperative in any fair election process as it ensures that all voices are heard and all communities are satisfied and engaged. 

In this paper we complement this notion by formalizing the opposite goal of preventing the \emph{overrepresentation} of any minority. We consider this to be a matter of security, and analyze a scenario where an adversarial minority may attempt to gain overrepresentation in the winning committee in order to capture the governance body or interfere with its correct functioning. 
Consequently, we consider the electoral system to be at risk of attack if the winning committee ever contains a subset of seats whose aggregate vote support is particularly low relative to the subset size, and establish an optimization problem that minimizes this risk. 
In this work we justify this problem from first principles and provide a thorough analysis of its computational complexity. 
%We also establish its connection to relevant axioms of proportional representation in the literature, 
We also study the performance of the most prominent election rules in the literature of proportional representation, the conclusion being that each of them either fails to provide a security guarantee or is prohibitively slow. Finally, we present a new, more efficient election rule that provides strong guarantees in terms of security as well as traditional proportional representation. Full details are provided below.

\paragraph{Blockchain architecture, verifiable computing and parallelism.}
Our work is motivated by an application on public, permissionless blockchain networks. 
These networks are equipped with computational and financial capabilities and have no central authority nor single point of failure, which gives them unprecedented levels of resistance to attacks, and for the first time removes the need for trusted intermediaries in peer-to-peer value transfers across the world. Recent years have seen an explosion of blockchain-based applications in finance, commerce, logistics, art and gaming; see~\cite{maesa2020blockchain} for a survey. 
Rather than controlling the identity and correct execution of each node, a blockchain network freely allows nodes to join or leave the network pseudonymously, and adds enough redundancy to resist the erroneous execution of any one of them. 
Nodes that participate in the consensus mechanism are known as \emph{validators}, and the correct functioning of the network is guaranteed as long as a supermajority of validators executes correctly.

Yet, the advantages of a blockchain architecture come at the cost of hard computational limitations. For every new block of transactions, all validators around the globe need to perform the same computations locally, and the network must wait for all of them to finish and agree on the results before processing the next block. Furthermore, a robust design should account for computationally weak validators such as consumer-grade computers, as imposing high performance requirements would lead to centralization, so the per-block computing load must remain conservatively low. Because of this, earlier networks such as Bitcoin~\cite{nakamoto2019bitcoin} and Ethereum 1.0~\cite{wood2014ethereum} can only process tens of transactions per second~\cite{chauhan2018blockchain}. 

As a consequence, only the most efficient algorithms -- such as those with a linear runtime -- can be directly implemented over a blockchain network. This represents a considerable hindrance for the use of committee election rules, in particular those sophisticated enough to provide guarantees on proportionality or security. 
As a case in point, the $\MMS$ rule~\cite{sanchez2016maximin} provides these guarantees, as we establish in this paper, but its slow (polynomial) runtime makes it unsuitable for implementation. On the other hand, the blockchain network EOS~\cite{griggeos} applies the multiwinner approval voting rule on its validator selection protocol, a rule that is highly efficient yet known to perform very poorly in terms of proportional representation; see~\cite[Table 2]{lackner2020approval}. 
The choice of this rule, in all likelihood based on operational considerations, has led to user discontent and claims of excessive centralization of the EOS network.%
\footnote{See the opinion piece ``EOS voting structure encourages centralization''~\cite{garg} as well as the 
news article ``Crypto ratings agency downgrades EOS for serious centralization problems''~\cite{chong}. 
Also, authors Br{\"u}njes et al.~\cite{brunjes2020reward} state ``At some point, controlling EOS delegates representing just 2.2\% of stakeholders was sufficient to halt the system... [and] only 8\% of total stake is represented by the 21 leading delegates.'' 
Finally, protocol designer Aarin Hagerty~\cite{hagerty} comments this about proposed election rules for EOS: ``Personally, I am not satisfied with any of these solutions because none of them addresses the proportional representation criterion. There are voting mechanisms that do satisfy it... The problem tends to be that they are actually very computationally intensive... For small numbers of winners and candidates it may be feasible (though this is still a difficult engineering challenge to build within smart contracts running on a blockchain), but if you go to even moderate numbers of winners and candidates it quickly explodes combinatorially and can even become infeasible to do off-chain... In summary, social choice theory is hard.''}
%We expand on \emph{validator selection protocols} further below.

A number of solutions have been proposed and sucessfully implemented for scaling up the computational capabilities of new-generation blockchain networks, such as sharding and layer-2 solutions; see~\cite{zhou2020solutions} for a survey. 
Of relevance to our work is the use of \emph{verifiable computing schemes}~\cite{gennaro2010non}. Such a scheme offloads a heavy task to one or more \emph{off-chain workers}, that is, entities that are logically separate from the rest of the network and may process the task on high-performance machines and with relaxed time frames as their operations do not interfere with block production. Once the task is completed and the output is fed back into the network, its correctness is verified by the validators. 
This is a sensible scheme for a task if its output can be subjected to a \emph{verification process} that a) guarantees correctness, even when the task is performed by an untrusted party, b) has a much lower runtime than performing the task itself from scratch, and c) admits \emph{parallelism}, so that it can be executed over multiple computing units, each with bounded time and memory complexities. 
These computing units may then be executed on consecutive blocks (sequentially) or separate shards (concurrently), depending on implementation.
One of our main contributions is showing that our proposed rule admits a verification process on its winning committee that can check the guarantees on proportionality and security in \emph{linear time} in the size of the input, i.e., the voters' ballots. Moreover, this verification can be executed over multiple computing units each in time \emph{linear in the number of candidates} and independent of the number of voters. 
In fact, the Polkadot network~\cite{burdges2020overview} is developing an implementation of our election rule within its validator selection protocol, as a verifiable computing scheme with parallelized verification, that can handle hundreds of candidates and a large number of voters; we include details of that protocol in Section~\ref{s:implement}. 

Our work thus constitutes an effort towards applying verifiable computing to election rules. 
Of course, one may argue that developing electoral systems that facilitate the public verification of results in accordance to clearly defined criteria is a worthwhile pursuit in itself, beyond any implementation concerns related to the blockchain architecture.
Yet, it is worth mentioning that the issue of implementability has become ever more relevant in recent years, and not only for blockchain-based solutions expressly built for online voting. New blockchain networks with any sort of functionality are likely to run elections in two of their core protocols. 
The first one is \emph{on-chain governance}~\cite{beck2018governance}: many projects are abandoning the notion of immutable code in favor of a more flexible design that facilitates future code upgrades via an embedded voting process of all holders of the native token. This process not only helps in terms of coordination but also legitimizes the result and avoids hard forks. 
Governance may also allow token holders to vote on committees %, councils 
and referenda, %launch their own candidacies, 
raise proposals, form commissions, etc. 
The second core protocol is validator selection, which we mentioned before and we describe in detail next as it is our motivating application and the background for our problem definition. 

\paragraph{Validator selection in Proof-of-Stake.}
Many blockchain networks launched in recent years substitute the highly inefficient Proof-of-Work (PoW) component of the consensus mechanism~\cite{nakamoto2019bitcoin} with Proof-of-Stake (PoS), in which the level of participation of validators depends on their token holdings --their stake-- as opposed to their computational power. 
While a pure PoS system allows any willing token holder to participate directly, most projects place a bound $k$ on the number of validators that can be active at any given moment. This bound may be set explicitly or implicitly, as a consequence of some barrier to entry. Arguments for setting such a bound are that the increase in operational costs and communication complexity eventually outmatches the marginal increase in benefits stemming from decentralization as $k$ grows, and that most users with little stake would find it inconvenient to keep a validator node constantly online for only sporadic participation, and would rather form validation pools in order to decrease the variance on their revenue and profit from economies of scale. %
Instead, a system may use ``representative democracy'' to formalize and facilitate the formation of these pools, allowing users to either launch the candidacy of their own nodes, or indicate the candidates that they trust. 
From this input, a committee with $k$ of the most trusted candidates emerges as the active validator set. Networks that broadly follow this approach include Polkadot~\cite{burdges2020overview}, Cardano~\cite{brunjes2020reward}, Tron~\cite{tron}, EOS~\cite{griggeos}, Cosmos~\cite{cosmos} and Tezos~\cite{tezos}, among many others. 

While similar in spirit, the approaches taken by these projects vary in several regards, most significantly in terms of incentives and the electoral system used. These design choices are of the utmost importance as they affect the decentralization and security levels achieved by the network; we refer again to the centralization issues experienced by EOS that we mentioned above. 
Yet, rigorous analyses behind these design choices are generally scarce. 
A notable exception is the recent work by Br{\"u}njes et al.~\cite{brunjes2020reward}, that proposes an incentive scheme for stake pools backed by a game theoretical analysis. In turn, in the present work we propose for the first time an electoral system for the selection of validators and analyze it from the perspective of computational social choice. 

We focus on Nominated Proof-of-Stake (NPoS), the design implemented by the Polkadot and Kusama networks~\cite{burdges2020overview}. In NPoS, any stakeholder is free to become a validator candidate, or a \emph{nominator} who provides an unranked list of candidates that she trusts. At regular intervals of a few hours, a committee of $k$ validators --in the order of hundreds-- is elected according to the current nominators' preferences. 
As a security measure, both validators and nominators have their stake locked as collateral, so that if a validator ever shows negligent or adversarial behavior, backing nominators are susceptible to losing their stake. Conversely, during normal execution the network provides economic rewards to all validators and their backing nominators in proportion to their stake and in a non-custodial way. Nominators are thus indirect participants in the consensus mechanism with a vested economic interest to guard the performance of validators and support only the most capable and trustworthy candidates. 
We provide further details about the NPoS mechanism in Section~\ref{s:implement}. 

\paragraph{Problem definition.}
For the sake of simplicity, in what follows we consider a model where only nominators have stake, not candidates, and we equate their stake amount to their voting strength. This leads to a vote-weighted, approval-based committee election problem. 
We remark that most of the following concepts are described in the literature in terms of unit votes; we generalize them to positive real valued vote strengths, following the principle that a voter with two units of strength is equivalent to two voters with unit strength and identical preferences. 
As mentioned before, we set to achieve both proportional representation and security. We formalize each of these goals next. 

\textbf{Proportional representation:} We aim to guarantee that nominators are not \emph{underrepresented} relative to their stake by the elected validators. 
We highlight that diverse preferences and factions may naturally arise among nominators for reasons that range from economically and technically motivated to political, geographical, etc., and that preserving this diversity among the elected validators ensures that the network stays decentralized. 

Electoral system designs that achieve some form of proportional representation have been present in the literature for a very long time. Of special note is the work of Scandinavian mathematicians Edvard Phragm\'{e}n and Thorvald Thiele in the late nineteenth century \cite{phragmen1894methode, phragmen1895proportionella, phragmen1896theorie, phragmen1899till, thiele1895om, janson2016phragmen}. 
Several axioms have been recently proposed to define the property mathematically -- we mention the most relevant ones. 
\emph{Justified representation} (JR)~\cite{aziz2017justified} states that if a group of voters is cohesive enough in terms of candidate preferences and has a large enough aggregate vote strength, then it has a justified claim to be represented by a member of the committee.
\emph{Proportional justified representation} (PJR)~\cite{sanchez2017proportional} says that such a group deserves not just one but a certain number of representatives in proportion to its vote strength, where a committee member is said to represent the group as long as it represents any voter in it.
Finally, \emph{extended justified representation} (EJR)~\cite{aziz2017justified} strengthens this last condition and requires not only that the group have enough representatives collectively, but that some voter in it have enough representatives individually.
It is known that EJR implies PJR and PJR implies JR, while converse implications are not true~\cite{sanchez2017proportional}. %
For each of these properties, a committee election rule is said to satisfy said property if its output committee always satisfies it for any input instance. 
While classical election rules usually achieve JR, they fail the stronger properties of PJR and EJR, and up to recently there were no known efficient rules that satisfy either of the latter two. 
For instance, the proportional approval voting (PAV) method \cite{thiele1895om, janson2016phragmen} proposed by Thiele satisfies EJR but is NP-hard to compute, while efficient heuristics based on it, such as reweighted approval voting, fail PJR \cite{aziz2014computational, skowron2016finding, aziz2017justified}. 
Only in the last five years have polynomial-time rules that achieve PJR or EJR finally been proposed \cite{brill2017phragmen, sanchez2016maximin, aziz2018complexity, peters2019proportionality}. 

Among these axioms, \textbf{we set to achieve PJR}, defined formally in Section~\ref{s:prel}, for two reasons. 
First, because it is more \emph{Sybil resistant}~\cite{douceur2002sybil} than JR. 
Concretely, under JR a strategic voter may be incentivized to spread her stake over multiple nominator identities in order to gain more representatives, but she has no such incentive under PJR.  
Second, because PJR seems to be most compatible with our security objective. Indeed, as claimed in~\cite{peters2019proportionality} and \cite{lackner2020approval}, the PJR and EJR axioms seem to correspond to different notions of proportionality: while EJR is primarily concerned with the voters' satisfaction, PJR considers proportionality of the voters' decision power, and our security objective aligns best with the latter notion. 
We establish the incompatibility between EJR and our security objective further below.

\textbf{Security:} 
As is the case in any PoS-based blockchain network, under NPoS the basic security assumption is that most of the stake is held by actors who behave honestly or rationally. Under this assumption, we consider an adversary that attempts to carry out an attack on the network, and has the power to create any number of identities including both nominators and candidates (via Sybil behavior), yet has a bounded stake budget. 
Depending on the type of attack, in order to succeed he will require that a minimum number of candidates under his control get elected in the committee, and he may recur to strategic voting to achieve this. Therefore, the security level corresponds to how difficult it is for a voter or group of voters with limited aggregate voting strength to gain \emph{overrepresentation} in the elected committee. 

Further formalizing our problem, we consider finite sets $N$ and $C$ of voters and candidates respectively, where every voter $n\in N$ provides a list $C_n\subseteq C$ of approved candidates and has a vote strength $s_n$. 
%There is also a target number $1\leq k< |C|$ of candidates to elect.
Suppose we want to make it as difficult as possible for an adversary to gain a certain threshold $1\leq r\leq k$ of representatives within the $k$-validator committee. 
Then, our goal would be to elect a committee $A\subseteq C$ that maximizes 
$$\min_{A'\subseteq A, \  |A'|=r} \quad \sum_{n\in N: \ C_n\cap A'\neq \emptyset} s_n.$$ 
For any subset $A'\subseteq A$ of $r$ seats in committee $A$, the quantity above is the aggregate vote strength that is backing any seat in $A'$. In our application, this quantity also corresponds to the total collateral susceptible to being lost if $A'$ carries out an attack; hence, maximizing this amount not only makes it difficult for the adversary to gain enough representatives, but also costly to attack if he does. Of course, on top of the potential loss of collateral, the adversary must also consider the potential loss of representation in future elections, which translates to loss of future payouts. %

We thus obtain a different optimization objective for each value of threshold $r$. 
If we are only concerned about a particular threshold, we can fix the corresponding objective. 
For example, for $r=1$, the objective is equivalent to the classical multiwinner approval voting rule: selecting the $k$ candidates $c\in C$ with highest total approval $\sum_{n\in N: \ c\in C_n} s_n$. 
Or, we could set $r$ to $\lceil k/3\rceil$ or to $\lceil k/2\rceil$, which are respectively the thresholds required to carry out a successful attack in classical Byzantine fault tolerant consensus~\cite{pease1980reaching} and in Nakamoto consensus~\cite{stifter2018agreement}. 
Yet, there are several types of attack vectors requiring different thresholds, and some attack attempts may succeed with higher probability when there are more attacking validators. %
For example, in a blockchain network that features a sharded architecture, validators may split into small, randomly generated \emph{commissions} in order to process multiple blocks in parallel, such that each block is subjected to several sequential approval rounds by different commissions. In particular, a threshold of, say, $5\%$ of validators may be enough to make an invalid block pass the first approval round, if all validators in the first commission are adversarial. 
While such an attack attempt still has a very low success probability thanks to the additional approval rounds by other randomly selected commissions, the security level is considerably stronger if one has a guarantee that any such attempt will be costly for the adversary, so that it makes no economic sense for a rational agent to ever try it. %
%\footnote{We remark that in this example, we begin with the assumption that active validators have equal power in consensus (or at least in the approval of shard blocks), and conclude that they need to have a somewhat uniform distribution of stake backings. It may be then natural to ask why not just give validators a power in consensus proportional to their stake backings. The problem with that solution is that a single adversarial validator controlling $5\%$ of the total stake could form a commission on its own, making it much easier for the adversary to attempt attacks.} 
Hence, a more pragmatic approach is to incorporate threshold $r$ into the objective and maximize \emph{the least possible cost per seat over all thresholds}, i.e.,  
\begin{align}\label{eq:security}
    \max_{A\subseteq C, \ |A|=k} \quad \min_{A'\subseteq A, \ A'\neq \emptyset} \quad \frac{1}{|A'|} \sum_{n\in N: \ C_n\cap A' \neq \emptyset} s_n.
\end{align}

We establish in Theorem~\ref{thm:equivalence} that this objective is equivalent to the \textbf{maximin support objective}, recently introduced by Sánchez-Fernández et al.~\cite{sanchez2016maximin}, which we thus set to optimize. 
We define it formally in Section~\ref{s:prel}.
%To define this last objective, which we do formally in Section~\ref{s:prel}, one needs the election rule to establish not only a winning committee $A\subseteq C$, but also a \emph{vote distribution}; that is, a fractional distribution of each voter $n$'s vote strength $s_n$ among her approved committee members in $C_n\cap A$.%
%\footnote{This is called a \emph{support distribution function} in~\cite{sanchez2016maximin}, and is related to the notion of a \emph{price system} in~\cite{peters2019proportionality}.} 
%For instance, for voter $n$ the election rule may assign a third of $s_n$ to $c_1$ and two thirds of $s_n$ to $c_2$, where $c_1, c_2\in C_n\cap A$. 
%The objective is then to maximize, over all possible committees and distributions, the least amount of vote assigned to any committee member. 
%We observe here that unlike most other applications of multiwinner elections, in NPoS there is practical utility in computing a vote distribution from nominators to the elected validators: by reversing its sense, it establishes the exact way in which the validators' payouts or penalties must be distributed back to the nominators.
The authors in~\cite{sanchez2016maximin} remark that in its exact version, maximin support is equivalent to another objective, $\maxphragmen$, devised by Phragm\'{e}n and recently analyzed in~\cite{brill2017phragmen}, and in this last paper it is shown that $\maxphragmen$ is NP-hard and incompatible with EJR. 
Thus, the same hardness and incompatibility with EJR holds true for our security objective. 
%To the best of our knowledge, the approximability of maximin support has not previously been studied.

\paragraph{Our contribution.}
Our security analysis for the selection of validators leads us to pursue the maximin support objective, which prevents overrepresentation. Conversely, we equate our proportionality goal to the PJR property, which prevents underrepresentation. 
We show that these goals are compatible and complement each other well, and prove the existence of efficient election rules that achieve guarantees on both of them. 

\begin{theorem}\label{thm:intro1}
There is an efficient rule $\phragmms$ for approval-based committee elections that simultaneously achieves the PJR property and a 3.15-approximation guarantee for maximin support.
\end{theorem}

To complement this result, we also provide a new hardness result for maximin support, which implies that the problem does not admit a PTAS%
\footnote{A \emph{polynomial time approximation scheme} (PTAS) for an optimization problem is an algorithm that, for any constant $\eps>0$ and any given instance, returns a $(1+\eps)$-factor approximation in polynomial time.}  unless P=NP.  
To the best of our knowledge, ours constitutes the first analysis of approximability not only for maximin support, but for any Phragm\'{e}n objective. 
In contrast, several approximation algorithms for Thiele objectives have been proposed; see~\cite{lackner2020approval} for a survey. 
Next comes the question of applicability: as mentioned previously, the blockchain architecture adds stringent constraints to computations. However, if the output can be \emph{verified} much faster than it can be computed from scratch, then the task can be implemented as a verifiable computing scheme. This is the case for our new election rule.

\begin{theorem}\label{thm:intro2}
There is a verification process that takes as input the election instance and an arbitrary solution to it, such that if the solution passes it then it is guaranteed to satisfy the PJR property and a 3.15-factor approximation for maximin support. 
Further, the output of the $\phragmms$ rule always passes this test.
Finally, the test has a runtime linear in the size of the input, and can be parallelized into multiple computing units each with a runtime linear in the number of candidates and independent of the number of voters. 
\end{theorem}

We remark that testing whether an arbitrary solution satisfies PJR is coNP-complete~\cite{aziz2018complexity}, and that passing the test above is a sufficient but not a necessary condition for a solution to have the two aforementioned properties; hence, the fact that the output of our proposed rule passes the test is not evident.
This result enables the first blockchain implementation of a validator selection protocol with strong theoretical guarantees on security and proportionality. 
%We provide further details for such a protocol in Section~\ref{s:implement}.
%
Finally, we derive from the new rule a post-computation which, when paired with any approximation algorithm for maximin support, makes it also satisfy PJR in a black-box manner.

\begin{theorem}\label{thm:intro3}
There is an efficient computation that takes as input an election instance and an arbitrary solution to it, and outputs a new solution which a) is no worse than the input solution in terms of the maximin support objective, b) satisfies the PJR property, and in particular c) can be verified to satisfy PJR in linear time.
\end{theorem}

This result shows that PJR is strongly compatible with maximin support (unlike EJR) and can be easily added to future approximation algorithms that may be developed for this objective.

\paragraph{Organization of the paper and technical overview.}
In the next section we formalize the objectives of our multiwinner election problem and provide required technical definitions, and in Section~\ref{s:complexity} we present a thorough complexity analysis for maximin support, including both new approximability and hardness results. 
We also compare the performance, relative to this objective, of the most relevant election rules in the literature of proportional representation. 
Our comparison provides new tools to discern between these rules. For instance, the survey paper~\cite{lackner2020approval} mentions $\phragmen$~\cite{brill2017phragmen} and $\MMS$~\cite{sanchez2016maximin} as two efficient rules that achieve the PJR property and leaves as an open question which of the two is preferable, whereas we show that out of the two only the latter provides a constant-factor approximation guarantee for maximin support. 

In Section~\ref{s:heuristic} we prove Theorem~\ref{thm:intro1} and present the new rule $\phragmms$, that takes inspiration from $\phragmen$ but has a more involved candidate selection heuristic that allows for better guarantees for both of our objectives. 
In Section~\ref{s:local} we prove Theorem~\ref{thm:intro2} and explore how guarantees for these objectives can be efficiently verified on the output solution. 
To do so, we define a parametric version of PJR, and link it to a notion of local optimality for our new rule, which is easy to test. 
In Section~\ref{s:LS} we transform our rule into a local search algorithm, and use it to prove Theorem~\ref{thm:intro3}. 
Finally, in Section~\ref{s:implement} we provide further details about NPoS and sketch a proposal for a validator selection protocol that implements $\phragmms$ as a verifiable computing scheme. 
We conclude in Section~\ref{s:conc}. 
%The proof of Theorem~\ref{thm:intro3} has been omitted from the present version of this paper, and appears in the full version~\cite{cevallos2020verifiably}; in that proof, we transform our rule into a local search algorithm.

In our analyses, we build upon the notion of \emph{load balancing} used in $\phragmen$, and consider distributions of votes from voters to committee members as a network flow over the bipartite approval graph.  
We define what an (ideally) \emph{balanced distribution} is, and in Appendix~\ref{s:balanced} we provide an algorithm to compute one efficiently for a fixed committee, using notions of parametric flow. 
Then, we synthesize the strategies of several election rules in the literature according to how well they balance vote distributions, and apply the flow decomposition theorem to derive approximation guarantees. Most of the proofs involving network flow theory are delayed to Appendix~\ref{s:flow}. 
For the sake of completeness, in Appendix~\ref{s:algorithms} we present algorithmic considerations to speed up the computation of the new $\phragmms$ rule, and in Appendix~\ref{s:lazymms} we show how one can shave off a factor $\Theta(k)$ from the runtime of the $\MMS$ rule~\cite{sanchez2016maximin} by using the theoretical tool set developed in this paper. Some delayed proofs are presented in Appendix~\ref{s:proofs}.

%In our analyses, we build upon the notion of \emph{load balancing} used in $\phragmen$, and consider distributions of votes from voters to candidates as a network flow over the bipartite approval graph.  
%We define what an (ideally) \emph{balanced distribution} is, and use it to synthesize the strategies of several heuristics in the literature according to how well they balance distributions. Finally, we apply the flow decomposition theorem to derive approximation guarantees. 

%The following results have been delayed to the full version of the paper~\cite{cevallos2020verifiably}: a) the proof of Theorem~\ref{thm:intro3}, where we transform our rule into a local search algorithm; b) an efficient algorithm to compute a balanced distribution for a fixed winning committee, using notions of parametric flow; c) most proofs involving network flow theory; d) optimizations to speed up the new $\phragmms$ rule; e) a variant of the $\MMS$ rule~\cite{sanchez2016maximin} that shaves off a factor $\Theta(k)$ from its runtime by using the theoretical tool set developed in this paper; and f) proofs of some minor lemmas.

%% file: preliminaries.tex
\section{Preliminaries}\label{s:prel}

Throughout the paper we consider the following approval-based multiwinner election instance. 
We are given a bipartite approval graph $G=(N\cup C, E)$ where $N$ is a finite set of voters and $C$ is a finite set of candidates. 
We are additionally given a vector $s\in\R^N$ of vote strengths, where $s_n$ is the strength of $n$'s vote, and a target number $k$ of candidates to elect, where $0< k<|C|$.
For each voter $n\in N$, $C_n:=\{c\in C: \ nc\in E\}$ represents her approval ballot, i.e., the subset of candidates that $n$ approves of, and for each candidate $c\in C$ we denote by $N_c:=\{n\in N: \ nc\in E\}$ the set of voters approving $c$, where $nc$ is shorthand for edge $\{n,c\}$. 
To avoid trivialities, we assume that the input graph $G$ has no isolated vertices. 
For any $c\in C\setminus A$, we write $A+c$ and $A-c$ as shorthands for $A\cup\{c\}$ and $A\setminus \{c\}$ respectively. 

\paragraph{Proportional justified representation.} 
The PJR property was introduced in~\cite{sanchez2017proportional} for voters with unit vote strength. We present its natural generalization to positive real valued strengths. 
A committee $A\subseteq C$ of $k$ members satisfies PJR if, for any group $N'\subseteq N$ of voters and any integer $0<r\leq k$, we have that 
\begin{itemize}
\item[a)] if $|\cap_{n\in N'} C_n|\geq r$
\item[b)] and $\sum_{n\in N'} s_n \geq r\cdot \hat{t}$, 
\item[c)] then $|A\cap (\cup_{n\in N'} C_n)|\geq r$,
\end{itemize}
where $\hat{t}:=\sum_{n\in N} s_n / k$. 
In words, if there is a group $N'$ of voters with at least $r$ commonly approved candidates, and enough aggregate vote strength to provide each of these candidates with a vote support of at least $\hat{t}$, then this group has a justified right to be represented by $r$ members in committee $A$, though not necessarily commonly approved. 
This right is justified because $\hat{t}$ is an upper bound on the average vote support that the full voter set $N$ can provide to any committee of $k$ members. 

\paragraph{Maximin support objective.} 
For the given instance, we consider a solution consisting of a tuple $(A,w)$, where $A\subseteq C$ is a committee of $k$ candidates, and $w\in\R^E$ is a vector of non-negative edge weights that represents a fractional distribution of each voter's vote among her approved candidates.%
\footnote{This edge weight vector is related to the notions of \emph{support distribution function} in~\cite{sanchez2016maximin} and \emph{price system} in~\cite{peters2019proportionality}. In particular, all election rules considered in this paper are \emph{priceable}, as defined in~\cite{peters2019proportionality}.} 
For instance, for a voter $n$ this distribution may assign a third of $s_n$ to $c_1$ and two thirds of $s_n$ to $c_2$, where $c_1, c_2\in C_n$.
Vector $w$ is \emph{feasible}%
\footnote{Intuitively, a feasible solution $(A,w)$ should also observe $w_{nc}=0$ for each edge $nc$ with $c\not\in A$. 
However, as this constraint can easily be enforced in post-computation, we ignore it so that the feasibility of a vector $w$ is independent of any committee.} 
if  % 
\begin{equation}
    \sum_{c\in C_n} w_{nc}\leq s_n \quad \text{ for each voter } n\in N. \label{eq:feasible}
\end{equation}

In our analyses, we will also consider \emph{partial} committees, with $|A|\leq k$. If $|A|=k$, we call it \emph{full}. 
All solutions $(A,w)$ in this paper are assumed to be feasible and full unless stated otherwise. 
Given a (possibly partial, unfeasible) solution $(A,w)$, we define the \emph{support} over the committee members as 
\begin{align}
\supp_w(c) &:=\sum_{n\in N_c} w_{nc} \quad \text{for each $c\in A, \quad$ and } \nonumber \\
\supp_w(A)& :=\min_{c\in A} \supp_w(c), \label{eq:support}
\end{align}
with the convention that $\supp_w(\emptyset)=\infty$ for any weight vector $w\in\R^E$. 
The maximin support objective, introduced in~\cite{sanchez2016maximin}, asks to maximize the least member support $\supp_w(A)$ over all feasible full solutions $(A,w)$. 

\paragraph{Balanced solutions.}
For a fixed committee $A$, a feasible weight vector $w\in\R^E$ that maximizes $\supp_w(A)$ can be found efficiently. 
However, we seek additional desirable properties for a weight vector which can still be achieved efficiently. We say that a feasible $w\in\R^E$ is \emph{balanced for $A$}, or that $(A,w)$ is a balanced solution, if
\begin{enumerate}
    \item it maximizes the sum of member supports, $\sum_{c\in A} \supp_w(c)$, over all feasible weight vectors, and 
    \item it minimizes the sum of supports squared, $\sum_{c\in A} \supp_w^2(c)$, over vectors that observe the property above. 
\end{enumerate}
In other words, a balanced weight vector maximizes the sum of supports and then minimizes their variance. 
%In the next lemma, whose proof is delayed to the full version~\cite{cevallos2020verifiably}, we establish some key properties that we exploit in our analyses.
In the next lemma, whose proof is delayed to Appendix~\ref{s:proofs}, we establish some key properties that we exploit in our analyses. 

\begin{lemma}\label{lem:balanced}
Let $(A,w)$ be a balanced, possibly partial solution. Then,
\begin{enumerate}
    \item vector $w$ simultaneously maximizes, for each $1\leq r \leq |A|$, the quantity $\min_{A'\subseteq A, |A'|=r} \ \sum_{c\in A'} \supp_{w'}(c)$ over all feasible weight vectors $w'\in\R^E$; 
		\item for each $n\in N$, $\sum_{c\in A\cap C_n} w_{nc}=s_n$ if $A\cap C_n\neq \emptyset$; and
    \item for each $n\in N$ and each candidate $c\in A\cap C_n$, if $w_{nc} > 0$ then $\supp_w(c)=\supp_w(A\cap C_n)$. 
\end{enumerate}
Furthermore, a feasible solution $(A,w)$ is balanced if and only if it observes properties 2 and 3 above, which can be tested in $O(|E|)$ time.
\end{lemma}

Notice that by setting $r=1$ in the first property, we obtain that balanced vector $w$ indeed maximizes the least member support $\supp_w(A)$ over all feasible weight vectors. 
More generally, for each threshold $r$ the quantity defined in the first property defines a lower bound on the cost for an adversary to get $r$ representatives in the validator committee in NPoS, so maximizing these objectives simultaneously for all thresholds $r$ aligns with our security objective as it makes any attack as costly as possible. 
The second point follows from the fact that the sum of member supports is maximal, so all the available vote strength must be distributed among candidates in $A$. 
The third point is a consequence of having the supports as evenly distributed as possible within $A$: for each voter $n$, all of her vote strength should be assigned exclusively to the candidates in $A\cap C_n$ with least support $\supp_w(A\cap C_n)$. 

In Appendix~\ref{s:balanced} we present algorithms for computing a balanced weight vector for a given committee $A$. 
%In the full version we present algorithms for computing a balanced weight vector for a given committee $A$.
In particular, we prove there that one can be found in time $O(|E|\cdot k + k^3)$ using parametric flow techniques, which to the best of our knowledge is the fastest algorithm in the literature even for the simpler problem of maximizing $\supp_w(A)$.

\begin{remark}\label{rem:bal}
In the remainder of the paper, we denote by $\bal$ the time complexity of finding a balanced weight vector for a given (possibly partial) committee, depending on the precise algorithm used.
\end{remark}

\begin{remark}
In all algorithms analyzed, we assume that all numerical operations take constant time.
\end{remark}

%% file: complexity.tex
\section{Analysis of maximin support}\label{s:complexity}

Consider a multiwinner election instance $(G=(N\cup C, E), s, k)$ as defined in Section~\ref{s:prel}. 
In this section we present a computational complexity analysis of the maximin support problem, including new results both on approximability and on hardness. 
We start by establishing that our security objective \eqref{eq:security} is indeed equivalent to maximin support. 
%The proof of the next theorem is delayed to the full version of the paper~\cite{cevallos2020verifiably}.
%The proof of the next theorem is delayed to Appendix~\ref{s:proofs}.

\begin{theorem} \label{thm:equivalence} 
For a fixed committee $A$, 
$$\max_{\text{feasible } w} \ \supp_w(A) = \min_{A' \subseteq A, \ A'\neq \emptyset} \ \frac{1}{|A'|} \sum_{n\in \cup_{c\in A'} N_c} s_n.$$
Hence maximin support, the problem of maximizing the left-hand side over all full committees $A$,  
is equivalent to that of maximizing the right-hand side over all full committees $A$. 
Furthermore, this equivalence preserves approximations, as any committee provides the same objective value to both problems.
\end{theorem}

%\begin{proof} %[Proof of Theorem~\ref{thm:equivalence}]
\begin{proof}
For a fixed committee $A$, let $w$ be a feasible edge weight vector that maximizes $\supp_w(A)$, and let $A'\subseteq A$ be the non-empty subset that minimizes the expression $\frac{1}{|A'|} \sum_{n\in \cup_{c\in A'} N_c} s_n$. Then, 
\begin{align*}
    \supp_w(A) &\leq \supp_w(A') \\
		& \leq \frac{1}{|A'|} \sum_{c\in A'} \supp_w(c) = \frac{1}{|A'|} \sum_{c\in A'} \sum_{n\in N_c} w_{nc} \\ 
    & = \frac{1}{|A'|}  \sum_{n\in \cup_{c\in A'} N_c} \quad \sum_{c\in C_n\cap A'} w_{nc} \\
    & \leq \frac{1}{|A'|} \sum_{n\in \cup_{c\in A'} N_c} s_n,
\end{align*}
where the second inequality follows by an averaging argument, and the last inequality follows by feasibility. 
This proves one inequality of the claim. 

To prove the opposite inequality, we assume for convenience and without loss of generality that $w$ is balanced for $A$. 
Let $A''\subseteq A$ be the set of committee members with least support, i.e., those $c\in A$ with $\supp_w(c)=\supp_w(A)$. Then,
\begin{align*}
    \supp_w(A) &= \supp_w(A'') = \frac{1}{|A''|} \sum_{c\in A''} \supp_w(c) \\
		& = \frac{1}{|A''|} \sum_{c\in A''} \sum_{n\in N_c} w_{nc} \\
    &= \frac{1}{|A''|} \sum_{n\in \cup_{c\in A''} N_c} \ \sum_{c\in C_n\cap A''} w_{nc} \\
    &= \frac{1}{|A''|} \sum_{n\in \cup_{c\in A''} N_c} \bigg( \sum_{c\in C_n\cap A} w_{nc} 
		- \sum_{c\in C_n \cap (A\setminus A'')} w_{nc}\bigg)\\
		&= \frac{1}{|A''|}\sum_{n\in \cup_{c\in A''} N_c} s_n 
		\geq \frac{1}{|A'|}\sum_{n\in \cup_{c\in A'}N_c} s_n,
\end{align*}
where in the last line we used the fact that for each voter $n\in \cup_{c\in A''} N_c$, the term $\sum_{c\in C_n\cap A} w_{nc}$ equals $s_n$ by property 2 of Lemma~\ref{lem:balanced}, while the term $\sum_{c\in C_n \cap (A\setminus A'')} w_{nc}$ vanishes by property 3 of Lemma~\ref{lem:balanced} and the definition of set $A''$. 
This proves the second inequality and completes the proof.
\end{proof}

\begin{figure}[htb]
  \centering
	\includegraphics[width=0.33\linewidth]{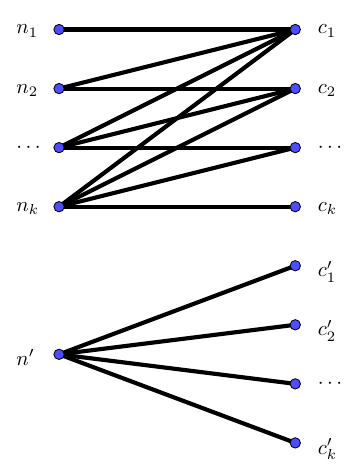}
  \caption{In this example, the last voter and last $k$ candidates are considered adversarial.}
  \label{fig:example}
\end{figure}

Next, we compare how some of the most relevant efficient election rules perform relative to a particular example, and highlight the effectiveness of the maximin support objective for preventing overrepresentation. 
Fix a committee size $k$ and consider an unweighted instance with $k+1$ voters and $2k$ candidates. 
For $1\leq i\leq k$, voter $n_i$ supports the candidate set $\{c_1, \cdots, c_i\}$, and the last voter $n'$ supports the candidate set $\{c_1', \cdots, c_k'\}$. See Figure~\ref{fig:example}.
We assume that an adversary controls the last voter and last $k$ candidates, and will use any elected representatives to disrupt the duties of the committee. 
How many representatives will he get? 

\begin{lemma}\label{lem:badexamples}
For any $\alpha \geq 1$, in the example above the number of adversarial candidates elected by a rule with an $\alpha$-approximation guarantee for maximin support is at most $\lfloor \alpha \rfloor$. 
On the other hand, this number is $\Omega(\sqrt{k})$ for proportional approval voting (PAV), and $\Omega(\log k)$ for both $\phragmen$ and Rule X. 
Hence, none of these three rules guarantees a constant-factor approximation. 
\end{lemma}

The proof is delayed to Appendix~\ref{s:proofs}. 
%The proof is delayed to the full version. 
For definitions of these rules, we direct the reader to the survey paper~\cite{lackner2020approval}.
We only remark here that Rule X~\cite{peters2019proportionality} is a recently proposed rule inspired in $\phragmen$ -- much as our own rule presented in the next section -- that achieves EJR. 
Thus, even election rules that achieve PJR and EJR may fare poorly in terms of our security objective.

The maximin support problem was introduced in~\cite{sanchez2016maximin}, where it was observed to be NP-hard. We show now a stronger hardness result for it, which in particular rules out the existence of a PTAS. %
%\footnote{A \emph{polynomial time approximation scheme} (PTAS) for an optimization problem is an algorithm that, for any constant $\eps>0$ and any given instance, returns a $(1+\eps)$-factor approximation in polynomial time.} 

\begin{theorem}
For any constant $\eps>0$, it is NP-hard to approximate the unweighted maximin support problem within a factor $\alpha=1.2-\eps$.
\end{theorem}

\begin{figure}[htb]
  \centering
  \includegraphics[width=0.7\linewidth]{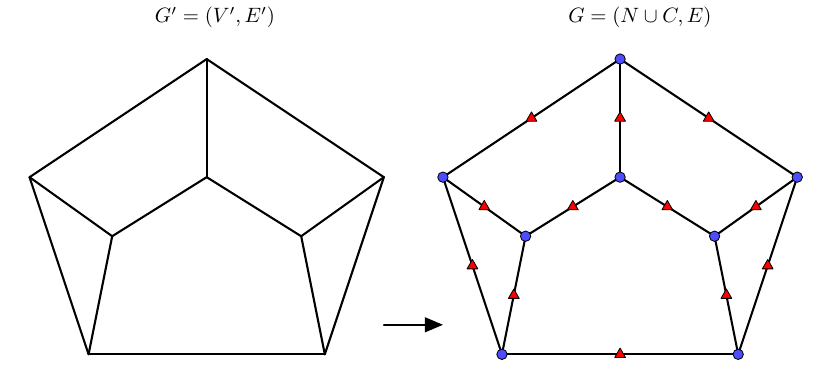}
  \caption{Reducing an instance of the $k$-independent set problem on cubic graphs to one of the maximin support problem. Set $N$ is represented by triangles and $C$ by circles.}
  \label{fig:hardness}
\end{figure}

\begin{proof}
We present a reduction from the $k$-independent set problem on cubic graphs, which is known to be NP-hard~\cite{johnson1979computers}. In this problem, one is given a graph $G'=(V',E')$ where every vertex has degree exactly 3, and a parameter $k'$, and one must decide whether there is a vertex subset $I\subseteq V'$ of size $k'$ such that no two vertices in $I$ are adjacent, i.e., $I$ is an independent set. 
Given such an instance, we define an instance $(G=(N\cup C, E), s, k)$ of maximin support where $k=k'$, $C=V'$ (each vertex in $V'$ corresponds to a candidate), and $N=E'$ with $s_n=1$ and $C_n=n$ for each $n\in N$ (each edge in $E'$ corresponds to a voter with unit vote that approves of the two candidates on its endpoints); see Figure~\ref{fig:hardness}.
Notice that in this instance, each candidate is approved by exactly 3 voters, and two candidates $c, c'$ have an approving voter in common if and only if $c$ and $c'$ are adjacent in $V'$.

Hence, if there is an independent set $I$ of size $k$ in $G'$, the same committee in $G$ can be assigned a vote distribution so that each member in it receives a support of 3 units, which is clearly maximal. 
On the other hand, if there is no independent set of size $k$ in $G'$, then for any solution $(A,w)$ of the maximin support instance there must be two committee members $c,c'\in A$ who have an approving voter in common. 
These two members have at most five voters approving either of them, so one of them must have a support of at most $5/2$. This shows that $\supp_w(A)\leq 5/2$ for any feasible solution $(A,w)$. 
Finally, the ratio between the objective values $3$ and $5/2$ is $6/5=1.2>\alpha$, so the assumed $\alpha$-approximation algorithm for maximin support would allow us to distinguish between these two cases and decide whether such an independent set $I$ exists. This completes the proof.
\end{proof}

In contrast, we show that the recently proposed $\MMS$ rule~\cite{sanchez2016maximin}, known to achieve the PJR property, also provides a 2-approximation for maximin support. 
In simple terms, $\MMS$ (Algorithm~\ref{alg:mms}) starts with an empty committee $A$ and iteratively adds candidates to it; in each iteration, it computes a balanced weight vector for each possible augmented committee that can be obtained by adding a candidate, and then inserts the candidate whose corresponding augmented committee has the highest least member support.

\begin{algorithm}[htb]
\SetAlgoLined
\KwData{Bipartite approval graph $G=(N\cup C, E)$, vector $s$ of vote strengths, target committee size $k$.}
Initialize $A=\emptyset$\ and $w=0\in \R^E$\;
\For{$i$ from $1$ to $k$}{
	\For{each candidate $c\in C\setminus A$}{
		Compute a balanced\footnotemark ~weight vector $w_c$ for $A+c$\;
		}
Find $c_i\in \arg\max_{c\in C\setminus A} \supp_{w_c}(A+c)$\;
Update $A\leftarrow A+c_i$ and $w\leftarrow w_{c_i}$\;
}
\Return $(A,w)$\;
\caption{$\MMS$, proposed in~\cite{sanchez2016maximin}}
\label{alg:mms}
\end{algorithm}
\footnotetext{The original algorithm in~\cite{sanchez2016maximin} does not compute a balanced weight vector, but any vector $w$ that maximizes $\supp_w(A)$, which is sufficient for our analysis. 
We consider balanced vectors here for ease of comparison with other algorithms in the paper and because this requirement does not seem to cause any increase in complexity.}

We will need the following key technical result, whose proof uses the flow decomposition theorem and is delayed to Appendix~\ref{s:flow}. 
%We will need the following key technical result, whose proof uses the flow decomposition theorem and is delayed to the full version.  
In simple terms, the lemma guarantees that for any partial solution, it is always possible to add a new candidate with a support of at least half the optimal maximin support value.

\begin{lemma}\label{lem:2sols}
If $(A^*, w^*)$ is an optimal solution to maximin support, and $(A,w)$ is a partial solution with $|A|\leq k$ and $A\neq A^*$, there is a candidate $c'\in A^*\setminus A$ and feasible solution $(A+c', w')$ such that 
$$\supp_{w'}(A+c')\geq \min\Big\{\supp_w(A), \frac{1}{2} \supp_{w^*}(A^*)\Big\}.$$
\end{lemma}

\begin{theorem}\label{thm:mms}
The $\MMS$ rule provides a 2-approximation guarantee for maximin support.
\end{theorem}

\begin{proof}
Let $(A_i, w_i)$ be the partial solution at the end of the $i$-th round of MMS, and let $(A^*, w^*)$ be an optimal solution. 
We prove by induction on $i$ that $\supp_{w_i}(A_i)\geq \frac{1}{2}\supp_{w^*}(A^*)$, where the case $i=0$ holds trivially as we use the convention that $\supp_w(\emptyset)=\infty$.
If the inequality holds for $i$, an application of Lemma~\ref{lem:2sols} for $(A_i, w_i)$ and $(A^*, w^*)$ implies that there is a candidate $c'\in A^*\setminus A_i$ and a feasible solution $(A_i+c', w')$ such that 
%\begin{align*}
%\supp_{w'}(A_i+c') &\geq \min\Big\{\supp_{w_i}(A_i), \frac{1}{2} \supp_{w^*}(A^*)\Big\} \
% &= \frac{1}{2} \supp_{w^*}(A^*).
%\end{align*}
%
$$\supp_{w'}(A_i+c') \geq \min\Big\{\supp_{w_i}(A_i), \frac{1}{2} \supp_{w^*}(A^*)\Big\} = \frac{1}{2} \supp_{w^*}(A^*).$$
As the algorithm is bound to inspect candidate $c'$ in round $i+1$, and compute for it a balanced weight vector $w_{c'}$ which maximizes the support of $A_i+c'$ (by Lemma~\ref{lem:balanced}), the solution $(A_{i+1}, w_{r+1})$ at the end of round $i+1$ must have an even higher support, i.e., %
%\begin{align*}
%\supp_{w_{i+1}}(A_{i+1}) &\geq \supp_{w_{c'}}(A_i+c) \\ 
%  &\geq \supp_{w'}(A_i+c) \geq \frac{1}{2} \supp_{w^*}(A^*).
%\end{align*}
%
$$\supp_{w_{i+1}}(A_{i+1}) \geq \supp_{w_{c'}}(A_i+c) \geq \supp_{w'}(A_i+c) \geq \frac{1}{2} \supp_{w^*}(A^*).$$
This completes the proof.
\end{proof}

$\MMS$ is a greedy algorithm with a runtime of $O(\bal \cdot |C|\cdot k)$, where we recall that $\bal$ is the time complexity of computing a balanced weight vector; see Remark~\ref{rem:bal}. 
To conclude the section we mention that a "lazy greedy" version of it can save a factor $\Theta(k)$ in the runtime while keeping the approximation guarantee virtually unchanged. 
%In the full version we prove the following result.
In Appendix~\ref{s:lazymms} we prove the following result.

\begin{theorem}\label{thm:2eps}
There is an algorithm $\lazy$ that, for any $\eps>0$, offers a $(2+\eps)$-approximation for the maximin support problem, satisfies the PJR property, and executes in time $O(\bal\cdot |C|\cdot \log(1/\eps))$.
\end{theorem}

%% file: heuristic.tex
\section{A new election rule}\label{s:heuristic}

The $\phragmen$ rule~\cite{brill2017phragmen} is highly efficient, with a runtime of $O(|E|\cdot k)$; see Algorithm~\ref{alg:phragmen} in Appendix~\ref{s:algorithms}. 
%The $\phragmen$ rule~\cite{brill2017phragmen} is highly efficient, with a runtime of $O(|E|\cdot k)$. 
However, as proved in the previous section, it fails to provide a good guarantee for the maximin support objective. 
On the other hand, $\MMS$~\cite{sanchez2016maximin} gives a constant-factor guarantee albeit with a slow running time that makes it unsuitable for implementation over a blockchain network.
In this section we introduce $\phragmms$, a new election rule inspired in $\phragmen$ that maintains a comparable runtime to it, yet lends itself to more robust analyses both for the maximin support objective and for the PJR property. 

\subsection{Inserting a candidate to a partial solution}\label{s:inserting}

We start with a brief analysis of the approaches taken in $\MMS$ and $\phragmen$. 
Both are iterative greedy algorithms that start with an empty committee and add to it a new candidate over $k$ iterations, following some specific heuristic for candidate selection.
For a given partial solution, $\MMS$ computes a balanced edge weight vector for each possible augmented committee resulting from adding one candidate, and keeps the one whose least support is largest. 
Naturally, such heuristic offers robust guarantees for maximin support but is slow as computing balanced vectors is costly. 
A similar approach is followed by $\phragmen$, except that it forgoes balancing vectors exactly. Instead, starting from the weight vector of the current committee, it rebalances it only approximately when a candidate is inserted, by performing local modifications in the neighborhood of the new candidate. 
Finally, $\phragmms$ follows the strategy of $\phragmen$ but uses a more involved heuristic for solution rebalancing, with a corresponding increase in runtime. 

In the algorithms described in this section we assume that there is a known background instance $(G=(N\cup C, E), s, k)$ that does not need to be passed as input. Rather, the input is a partial solution $(A,w)$ with $|A|\leq k$. We also assume that the current list of committee member supports $(\supp_w(c))_{c\in A}$ is implicitly passed by reference and updated in every algorithm.

Let $c'\in C\setminus A$ be a candidate that we consider adding to $(A,w)$. To do so, we modify weight vector $w$ into a new feasible vector $w'$ that redirects towards $c'$ some of the vote strength of the approving voters in $N_{c'}$, in turn decreasing the support of the current committee members that are also approved by these voters. Now, for a given threshold $t\geq 0$, we want to make sure not to reduce the support of any member $c$ below $t$, assuming it starts above $t$, and not to reduce it at all otherwise. A simple rule to ensure this is as follows: for each voter $n$ in $N_{c'}$ and each member $c\in A\cap C_n$, reduce the weight on edge $nc$ from $w_{nc}$ to $w_{nc}\cdot \min\{1, t/\supp_w(c)\}$, and assign the difference to edge $nc'$. That way, even if all edges incident to $c$ are so reduced in weight, the support of $c$ is scaled by a factor no smaller than $\min\{1, t/\supp_w(c)\}$ and hence its support does not fall below $t$.
Therefore, if for each voter $n\in N$ and threshold $t\geq 0$ we define that voter's \emph{slack} as

\begin{align}
    \slack_{(A,w)}(n,t):= s_n - \sum_{c\in A\cap C_n} w_{nc} \cdot\min \Big\{ 1, t/\supp_w(c)\Big\} \label{eq:slack}
\end{align}
and for each unelected candidate $c'\in C\setminus A$ and threshold $t\geq 0$ we define that candidate's \emph{parameterized score} as
\begin{equation}\label{eq:prescore}
    \prescore_{(A,w)}(c',t) := \sum_{n\in N_{c'}} \slack_{(A,w)}(n,t),
\end{equation}
then we can add $c'$ to the current partial solution with a support of $\prescore_{(A,w)}(c',t)$, while not making any other member's support decrease below threshold $t$. The resulting weight modification rule is formalized in Algorithm~\ref{alg:ins}. The next lemma easily follows from the previous exposition and its proof is skipped.

\begin{algorithm}[htb]
\SetAlgoLined
\KwData{Partial feasible solution $(A,w)$, candidate to insert $c'\in C\setminus A$, threshold $t\geq 0$.}
Initialize the new weight vector $w'\leftarrow w$\;
\For{each approving voter $n\in N_{c'}$}{
Set $w'_{nc'} \leftarrow s_n$\;
\For{each current member $c\in A\cap C_n$}{
\If{$\supp_w(c)>t$}{
	Update $w'_{nc} \leftarrow w'_{nc}\cdot\frac{t}{\supp_w(c)}$\;
}
Update $w'_{nc'}\leftarrow w'_{nc'} - w'_{nc}$\;
}
}
\Return $(A+c',w')$\;
 \caption{$\ins(A,w,c',t)$}
\label{alg:ins}
\end{algorithm}

\begin{lemma}\label{lem:insert}
For a feasible partial solution $(A,w)$, candidate $c'\in C\setminus A$ and threshold $t\geq 0$, 
Algorithm $\ins(A,w,c',t)$ executes in time $O(|E|)$ and returns a feasible partial solution $(A+c',w')$ 
such that $\supp_{w'}(c)\geq \min\{\supp_w(c),t\}$ for each member $c\in A$, and $\supp_{w'}(c')=\prescore_{(A,w)}(c',t)$. 
%In particular, if $\prescore_{(A,w)}(c',t)\geq t$ then $\supp_{w'}(A+c')\geq \min\{\supp_w(A),t\}$.
\end{lemma}

Whenever partial solution $(A,w)$ is clear from context, we drop the subscript from our notation of slack and parameterized score. %
Parameter $t$ provides a trade-off between the amount of support we direct to the new candidate $c'$ and the support we leave for the current members. We balance this trade-off by selecting the largest possible $t$ for which the inequality $\prescore(c',t)\geq t$ holds.
Thus, for each unelected candidate $c'\in C\setminus A$ we define its \emph{score} as 
\begin{align}
    \score_{(A,w)}(c'):=\max\{t\geq 0: \ \prescore_{(A,w)}(c',t)\geq t\},
\end{align}
where once again we drop the subscript if $(A,w)$ is clear from context. Our heuristic now becomes apparent.

\begin{heuristic}
For a partial solution $(A,w)$, find a candidate $c_{\max}\in C\setminus A$ with highest score $t_{\max}=\max_{c'\in C\setminus A} \score(c')$, and execute $\ins(A,w,c_{\max},t_{\max})$ so that for the new solution $(A+c_{\max},w')$: 
\begin{align*}
\forall c\in A, \ \supp_{w'}(c) &\geq \min\{\supp_w(c), t_{\max}\}, \quad \text{ and } \\
 \supp_{w'}(A+c_{\max}) &\geq \min \Big\{ \supp_w(A), t_{\max}\Big\}.
\end{align*}
\end{heuristic}

In Appendix~\ref{s:algorithms} we describe efficient algorithms to find the candidate with highest parameterized score for a given threshold $t$, as well as the candidate with overall highest score.
%In the full version~\cite{cevallos2020verifiably} we describe efficient algorithms to find the candidate with highest parameterized score for a given threshold $t$, as well as the candidate with overall highest score.

\begin{theorem}\label{thm:runtimes}
For a partial solution $(A,w)$ and threshold $t\geq 0$, there are algorithms $\maxprescore(A,w,t)$, that runs in time $O(|E|)$ and returns a tuple $(c_t,p_t)$ with $c_t\in C\setminus A$ and $p_t=\prescore(c_t,t)=\max_{c'\in C\setminus A} \prescore(c',t)$, 
and $\maxscore(A,w)$, that executes in time $O(|E|\cdot \log k)$ and returns a tuple $(c_{\max}, t_{\max})$ such that $c_{\max}\in C\setminus A$ and $t_{\max}=\score(c_{\max})=\max_{c'\in C\setminus A} \score(c')$.
\end{theorem}

Our heuristic for candidate selection, which finds a candidate with highest score and adds it to the current partial solution (Algorithm $\maxscore$ followed by $\ins$) executes in time $O(|E|\cdot \log k)$. 
It thus matches up to a logarithmic term the running time of the $\phragmen$ heuristic which is $O(|E|)$ per iteration. 
%In the full version we draw further parallels between $\phragmen$ and the new heuristic, and explain how the latter can be seen as a natural complication of the former that always grants higher scores to candidates and thus inserts them with higher supports.
In Appendix~\ref{s:algorithms} we draw further parallels between $\phragmen$ and the new heuristic, and explain how the latter can be seen as a natural complication of the former that always grants higher scores to candidates and thus inserts them with higher supports.

%% file: approx315.tex
\subsection{Inserting and rebalancing iteratively}\label{s:315}

We proved in Section~\ref{s:complexity} the existence of a 2-approximation algorithm for maximin support that runs in time $O(\bal\cdot |C|\cdot k)$ or $O(\bal\cdot |C|)$ (Theorems \ref{thm:mms} and \ref{thm:2eps} respectively). 
We use now our heuristic to develop $\phragmms$, a $3.15$-approximation algorithm that runs in time $O(\bal\cdot k)$ and satisfies PJR as well. 
We highlight that this is the fastest known election rule to achieve a constant-factor guarantee for maximin support, and that gains in speed are of paramount importance for our blockchain application, where there are hundreds of candidates and a large number of voters.

$\phragmms$ (Algorithm~\ref{alg:balanced}) is an iterative greedy algorithm that starts with an empty committee and alternates between inserting a new candidate with the new heuristic, and fully rebalancing the weight vector, i.e., replacing it with a balanced one. This constitutes a middle ground between the approach in $\phragmen$ where a balanced vector is never computed, and the approach in $\MMS$ where $O(|C|)$ balanced vectors are computed per iteration. 
We formalize the procedure below. Notice that running the Insert procedure (Algorithm~\ref{alg:ins}) before rebalancing is optional, but the step simplifies the analysis and may provide a good starting point to the balancing algorithm.

\begin{algorithm}[htb]
\SetAlgoLined
\KwData{Approval graph $G=(N\cup C, E)$, vector $s$ of vote strengths, committee size $k$.}
Initialize $A=\emptyset$\ and $w=0\in\R^E$\;
\For{$i$ from $1$ to $k$}{
Let $(c_{\max},t_{\max})\leftarrow \maxscore(A,w)$ \; 
\tcp{candidate w.~highest score, and its score} 
Update $(A,w)\leftarrow \ins(A,w,c_{\max},t_{\max})$ \; 
\tcp{or optionally just update $A\leftarrow A+c_{\max}$} 
Replace $w$ with a balanced weight vector for $A$\;
}
\Return $(A,w)$\;
\caption{$\phragmms$}
\label{alg:balanced}
\end{algorithm}

\begin{theorem}\label{thm:315}
$\phragmms$ offers a $3.15$-approximation guarantee for the maximin support problem, satisfies the PJR property, and runs in time $O(\bal\cdot k)$, assuming that $\bal= \Omega(|E|\cdot \log k)$.
\end{theorem}

This in turn proves Theorem~\ref{thm:intro1}. 
The claim on runtime is straightforward: we established in Theorem~\ref{thm:runtimes} that $\maxscore$ runs in time $O(|E|\cdot \log k)$, so each iteration of $\phragmms$ has a runtime of $O(|E|\cdot \log k + \bal)=O(\bal)$, assuming that $\bal= \Omega(|E|\cdot \log k)$. 
In fact, in Appendix~\ref{s:algorithms} we improve upon this analysis and show how each iteration can run in time $O(|E| + \bal)$.
%In fact, in the full version we improve upon this analysis and show how each iteration can run in time $O(|E| + \bal)$.
Next, in order to prove the PJR property we need the following technical lemmas.

\begin{lemma}\label{lem:2balanced}
If $(A,w)$ and $(A',w')$ are two balanced partial solutions with $A\subseteq A'$, then $\supp_w(c) \geq \supp_{w'}(c)$ for each $c\in A$, and $\score_{(A,w)}(c')\geq \score_{(A',w')}(c')$ for each $c'\in C\setminus A'$.
\end{lemma}

\begin{lemma}\label{lem:localopt}
If the inequality 
\begin{equation} \label{eq:localopt}
\supp_w(A)\geq \max_{c'\in C\setminus A} \score(c')
\end{equation}  
holds for a full solution $(A,w)$, then $A$ satisfies PJR.
\end{lemma}

Lemma \ref{lem:2balanced} formalizes the intuition that as more candidates are added to a partial solution that is kept balanced, the scores of unelected candidates may only decrease, never increase; its proof is delayed to Appendix~\ref{s:proofs}. 
%Lemma \ref{lem:2balanced} formalizes the intuition that as more candidates are added to a partial solution that is kept balanced, the scores of unelected candidates may only decrease, never increase; its proof is delayed to the full version.
Lemma~\ref{lem:localopt} establishes the key connection that exists between our definition of score -- and by extension our heuristic -- and the PJR property, and its proof is delayed to the end of Section~\ref{s:local}. 
We prove now that the output of $\phragmms$ satisfies inequality~\ref{eq:localopt}, and hence satisfies PJR.

\begin{lemma}\label{lem:315localoptimality}
At the end of each one of the $k$ iterations of Algorithm $\phragmms$, if $(A,w)$ is the current partial balanced solution, we have that $\supp_{w}(A)\geq \max_{c'\in C\setminus A} \score_{(A,w)}(c')$.
\end{lemma}
\begin{proof}
Let $(A_i,w_i)$ be the partial solution at the end of the $i$-th iteration. We prove the claim by induction on $i$, with the base case $i=0$ being trivial as we use the convention that $\supp_{w_0}(\emptyset)=\infty$ for any $w_0$. For $i\geq 1$, suppose that on iteration $i$ we insert a candidate $c_i$ with highest score, and let $w'$ be the vector that is output by $\ins(A_{i-1}, w_{i-1}, c_i, \score_{(A_{i-1}, w_{i-1})}(c_i))$ (Algorithm~\ref{alg:ins}). Then

\begin{align*}
\supp_{w_i}(A_i) &\geq \supp_{w'}(A_i) \\
&\geq \min\{ \supp_{w_{i-1}}(A_{i-1}), \score_{(A_{i-1}, w_{i-1})}(c_i) \} \\
&\geq \max_{c'\in C\setminus A_{i-1}} \score_{(A_{i-1}, w_{i-1})}(c') \\
&\geq \max_{c'\in C\setminus A_{i}} \score_{(A_{i}, w_{i})}(c'), 
\end{align*}
where the first inequality holds as $w_i$ is balanced for $A_i$, the second one is a property of our heuristic, the third one holds by induction hypothesis and the choice of candidate $c_i$, and the last one follows from Lemma~\ref{lem:2balanced}. 
\end{proof}

It remains to prove the claimed approximation guarantee for $\phragmms$. 
To do that, we use the following key technical result, whose proof is based on the flow decomposition theorem and is delayed to Apendix~\ref{s:flow}.
%To do that, we use the following key technical result, whose proof uses the flow decomposition theorem and is delayed to the full version. 
This result says that for a balanced partial solution, not only are there unelected candidates that can be appended with high support (as is the statement of Lemma~\ref{lem:2sols}), but they also have large scores, so we can find them efficiently with our heuristic. 
To this end, we show that there is a subset of voters with large aggregate vote strength and few representatives, so all their representatives all have large supports, and in turn the voters have large available slack.   

\begin{lemma}\label{lem:N_a}
If $(A^*, w^*)$ is an optimal solution to the maximin support instance with $t^*=\supp_{w^*}(A^*)$, and $(A,w)$ is balanced with $|A|\leq k$ and $A\neq A^*$, then for each $0\leq a\leq 1$ there is a subset $N(a)\subseteq N$ of voters such that 
\begin{enumerate}
	\item each voter $n\in N(a)$ approves of a candidate in $A^*\setminus A$;
	\item for each voter $n\in N(a)$, we have $\supp_w(A\cap C_n)\geq at^*$;
	\item $\sum_{n\in N(a)} s_n \geq |A^* \setminus A|\cdot (1-a) t^*$; and
	\item for any $b$ with $a\leq b\leq 1$ we have that $N(b)\subseteq N(a)$, and a voter $n\in N(a)$ belongs to $N(b)$ if and only if $n$ observes property 2 above with parameter $a$ replaced by $b$.
\end{enumerate}
\end{lemma}

As a warm-up, we show how this last result easily implies a $4$-approximation guarantee for $\phragmms$.

\begin{lemma}
If $(A,w)$, $(A^*,w^*)$ and $t^*$ are as in Lemma~\ref{lem:N_a}, there is a candidate $c'\in A^*\setminus A$ with $\score(c')\geq t^*/4$. Hence, $\phragmms$ provides a $4$-approximation for the maximin support problem.
\end{lemma}

\begin{proof}
We apply Lemma~\ref{lem:N_a} with $a=1/2$, and we refer to the four properties stated in that lemma. We have that

\begin{align*}
    \sum_{c'\in A^*\setminus A} \prescore(c',t^*/4) &=\sum_{c'\in A^*\setminus A} \sum_{n\in N_{c'}} \slack(n,t^*/4) \\
		&\geq \sum_{n\in N(a)} \slack(n,t^*/4) \\
    &\geq \sum_{n\in N(a)} \Big[ s_n - \frac{t^*}{4}\sum_{c\in A\cap C_n} \frac{w_{nc}}{\supp_w(c)} \Big] \\
    &\geq \sum_{n\in N(a)} \Big[ s_n - \frac{1}{2}\sum_{c\in A\cap C_n} w_{nc} \Big] \\
    &\geq \frac{1}{2}\sum_{n\in N(a)} s_n \\
		&\geq \frac{1}{2} (|A^*\setminus A|\cdot t^*/2) = |A^*\setminus A|\cdot t^*/4, 
\end{align*}
where the five inequalities hold respectively by property 1 (which implies $N(a)\subseteq \cup_{c'\in A^*\setminus A} N_{c'}$), by definition of slack (equation~\ref{eq:slack}), by property 2, by feasibility (inequality~\ref{eq:feasible}), and by property 3.
Therefore, by an averaging argument, there must be a candidate $c'\in A^*\setminus A$ with $\prescore(c',t^*/4)\geq t^*/4$, which in turn implies that $\score(c')\geq t^*/4$ by definition of score. 
The $4$-approximation guarantee for Algorithm $\phragmms$ easily follows by induction on the $k$ iterations, using Lemma~\ref{lem:insert} and the fact that rebalancing a partial solution never decreases its least member support.
\end{proof}

To get a better approximation guarantee for the $\phragmms$ rule and finish the proof of Theorem~\ref{thm:315}, we apply Lemma~\ref{lem:N_a} with a more carefully selected parameter $a$, 
and use the following technical result whose proof is delayed to Appendix~\ref{s:proofs}.
%and use the following technical result whose proof is delayed to the full version.

\begin{lemma}\label{lem:Lebesgue}
Consider a strictly increasing and differentiable function $f:\mathbb{R}\rightarrow \mathbb{R}$, with a unique root $\chi$. For a finite sum $\sum_{i\in I} \alpha_i f(x_i)$ where $\alpha_i\in\mathbb{R}$ and $ x_i\geq \chi$ for each $i\in I$, we have that
$$\sum_{i\in I} \alpha_i f(x_i) = \int_{\chi}^{\infty} f'(x) \big(\sum_{i\in I: \ x_i\geq x} \alpha_i\big)dx.$$
\end{lemma}

\begin{lemma}\label{lem:candidate315}
If $(A,w)$, $(A^*,w^*)$ and $t^*$ are as in Lemma~\ref{lem:N_a}, there is a candidate $c'\in A^*\setminus A$ with $\score(c')\geq t^*/3.15$. Hence, $\phragmms$ provides a $3.15$-approximation for the maximin support problem.
\end{lemma}

\begin{proof}
We refer to Lemma~\ref{lem:N_a} and its properties, with a parameter $0\leq a\leq 1$ to be defined later. We have
\begin{align*}
    \sum_{c'\in A^*\setminus A} \prescore(c',at^*)  
		&= \sum_{c'\in A^*\setminus A} \ \sum_{n\in N_c} \slack(n, at^*) \\
		&\geq \sum_{n\in N(a)} \slack(n, at^*) \\
    &\geq \sum_{n\in N(a)} \Big[ s_n - at^* \sum_{c\in A\cap C_n} \frac{w_{nc}}{\supp_w(c)} \Big] \\
    &\geq \sum_{n\in N(a)} \Big[ s_n - \frac{at^*}{\supp_w(A\cap C_n)} \sum_{c\in A\cap C_n} w_{nc} \Big] \\
    &\geq \sum_{n\in N(a)} s_n\Big[ 1- \frac{at^*}{\supp_w(A\cap C_n)} \Big], 
\end{align*}
where the four inequalities hold respectively by property 1, equation~\ref{eq:slack}, property 2 and inequality~\ref{eq:feasible}, and where $\supp_w(\emptyset)=\infty$ by convention. 
At this point, we apply Lemma~\ref{lem:Lebesgue} over function $f(x):=1-a/x$, which has the unique root $\chi=a$, and index set $I=N(a)$ with $\alpha_n=s_n$ and $x_n=\supp_w(A\cap C_n)/t^*$. We obtain
\begin{align*}
    \sum_{c'\in A^*\setminus A} \prescore(c',at^*) 
		&\geq \int_{a}^{\infty} f'(x) \Big( \sum_{n\in N(a): \ \supp_w(A\cap C_n)\geq xt^*} s_n \Big)dx\\
    &=\int_{a}^{\infty} \frac{a}{x^2}\Big( \sum_{n\in N(x)} s_n \Big)dx \\
    &\geq \int_{a}^1 \frac{a}{x^2} \Big( |A^*\setminus A|\cdot (1-x)t^* \Big)dx \\
    & = |A^*\setminus A|\cdot at^* \int_{a}^1 \Big( \frac{1}{x^2} - \frac{1}{x} \Big)dx \\
		&= |A^*\setminus A|\cdot at^*\Big(\frac{1}{a} - 1 + \ln  a\Big),
\end{align*}
where we exploited properties 4 and 3. 
If we now set $a=1/3.15$, we have that $1/a - 1 + \ln a\geq 1$, so by an averaging argument there is a candidate $c'\in A^*\setminus A$ for which $\prescore(c',at^*)\geq at^*$, and hence $\score(c')\geq at^*$. The approximation guarantee for the $\phragmms$ rule follows by induction on the $k$ iterations, as before.
\end{proof}

%% file: verifiable.tex
\section{Verifying the solution}\label{s:local}

We start the section with a key property of algorithm $\phragmms$ as motivation.

\begin{theorem}\label{thm:315guarantee}
If a balanced solution $(A,w)$ observes inequality~\ref{eq:localopt}, then it simultaneously satisfies PJR and a 3.15-approximation for maximin support. 
Testing all conditions (feasibility, balancedness and the previous inequality) can be done in time $O(|E|)$, and the output solution of $\phragmms$ is guaranteed to satisfy these conditions.
\end{theorem}

\begin{proof}
The first statement follows from Lemmas \ref{lem:localopt} and \ref{lem:candidate315}, and the third one from Lemma~\ref{lem:315localoptimality}. 
Feasibility (inequality~\ref{eq:feasible}) can clearly be checked in time $O(|E|)$, as can balancedness by Lemma~\ref{lem:balanced}. 
If $t:=\supp_w(A)$, inequality~\ref{eq:localopt} is equivalent to $t\geq \max_{c'\in C\setminus A} \prescore(c',t)$, which is tested with algorithm $\maxprescore(A,w,t)$ in time $O(|E|)$ by Theorem~\ref{thm:runtimes}.
\end{proof}

As we argued in the introduction, the result above is one of the most relevant features of our proposed election rule, and is essential for its implementation over a blockchain network as it enables its adaptation into a verifiable computing scheme. As such, the rule may be executed by off-chain workers, leaving only the linear-time tests mentioned in the previous theorem to be performed on-chain, to ensure the quality of the solution found. 
   
To finish the proof of Theorem~\ref{thm:intro2}, it remains to prove Lemma~\ref{lem:localopt} -- which we do at the end of this section -- and show that the verification process above admits a parallel execution -- which we do next. 
In particular, for a parameter $p$, we consider the distribution of this process over $p$ computing units that execute in sequence, such as $p$ consecutive blocks in a blockchain network. 
We remark however that our description below may be easily adapted to concurrent execution if desired.

\begin{lemma}\label{lem:parallel}
For any integer $p\geq 1$, both the input election instance as well as any solution $(A,w)$ to it can be distributed into $p$ data sets, such that each data set is of size $O(|E|/p + |C|)$. 
Moreover, all the tests mentioned in Theorem~\ref{thm:315guarantee} can be executed by $p$ sequential computing units such that each unit only requires access to one data set and runs in time $O(|E|/p + |C|)$. 
Therefore, for $p$ sufficiently large, each unit can be made to run in time $O(|C|)$.
\end{lemma}

\begin{proof}
Partition the voter set $N=\cup_{i=1}^p N^i$ into $p$ subsets of roughly equal size, and let $G^i$ be the subgraph of the input approval graph $G$ induced by $N^i\cup C$ (corresponding to the ballots of voters in $N^i$). 
Consider $p$ data sets where the $i$-th one stores subgraph $G^i$ along with the list of vote strengths for voters in $N^i$. 
Next, we assume that an untrusted party provides a solution $(A,w)$, and we assume that they also provide its corresponding vector $(\supp'_w(c))_{c\in A}$ of member supports, where the prime symbol indicates that these are claimed values to be verified. 
This solution is distributed so that the $i$-th data set stores the full committee $A$, the claimed supports, and the restriction $w^i$ of the edge weight vector $w$ over $G^i$. 
Clearly, each data set is of size $O(|E|/p + |C|)$.

Now consider $p$ computing units running in sequence, where the $i$-th unit has access to the $i$-th data set, and recall that the verification of solution $(A,w)$ consists of four tests: a) feasibility, b) balancedness, c) correctness of the claimed member supports, and d) the inequality $t\geq \max_{c'\in C\setminus A} \prescore(c',t)$, where we define $t:=\supp'_w(A)$. 
To avoid dependencies across these tests, our general strategy is to assume that the claimed supports are correct, except obviously for test c. 
For example, since both feasibility (inequality~\ref{eq:feasible}) and balancedness (properties 2 and 3 of Lemma~\ref{lem:balanced}) are checked on a per-voter basis, the $i$-th unit can perform these checks for its own subset of voters $N^i$, using the claimed supports to check property 3 of Lemma~\ref{lem:balanced}. 

Tests c and d, on the other hand, require the cooperation of all units. 
The $i$-th unit can compute a vector $(\supp_{w^i}(c))_{c\in A}$ of supports relative to the local voters in $N^i$, so it follows by induction that it can also compute the partial sum $\sum_{j\leq i} (\supp_{w^j}(c))_{c\in A}$, and communicate it to the $(i+1)$-st unit. The last unit can then compute the full vector $\sum_{j\leq p} (\supp_{w^j}(c))_{c\in A}=(\supp_{w}(c))_{c\in A}$, and check that it matches the claimed supports. 

Similarly, the $i$-th unit can compute the vector $(\slack(n,t))_{n\in N^i}$ of slacks for $N^i$ (equation~\ref{eq:slack}), and use it to find a vector of parameterized scores relative to the local voters, $(\prescore^i(c',t))_{c'\in C\setminus A}$, where $\prescore^i(c', t):=\sum_{n\in N_{c'}\cap N^i} \slack(n,t)$. 
Again by induction, this unit can also find the partial sum $\sum_{j\leq i} (\prescore^j(c',t))_{c'\in C\setminus A}$, and communicate it to the $(i+1)$-st unit. 
The last unit then retrieves the full vector $\sum_{j\leq p} (\prescore^j(c',t))_{c'\in C\setminus A}=(\prescore(c',t))_{c'\in C\setminus A}$, and verifies that all parameterized scores are bounded by $t$. 
Clearly, each unit has a time and memory complexity of $O(|E|/p + |C|)$.
\end{proof}

We begin our analysis of the PJR property by defining a \emph{parametric version} of it, that measures just how well represented the voters are by a given committee. This generalization turns the property from binary to quantitative.

\begin{definition}
For any $t\in\R$, a committee $A\subseteq C$ (of any size) satisfies PJR with parameter $t$ ($t$-PJR for short) if, for any group $N'\subseteq N$ of voters and any integer $0<r\leq |A|$, we have that
\begin{itemize}
\item[a)] if $|\cap_{n\in N'} C_n|\geq r$
\item[b)] and $\sum_{n\in N'} s_n \geq r\cdot t$, 
\item[c)] then $|A\cap (\cup_{n\in N'} C_n)|\geq r$.
\end{itemize}
\end{definition}

In words, if there is a group $N'$ of voters with at least $r$ commonly approved candidates, and enough aggregate vote strength to provide each of these candidates with a support of at least $t$, then this group must be represented by at least $r$ members in committee $A$, though not necessarily commonly approved. 
Notice that the standard version of PJR is equivalent to $\hat{t}$-PJR for $\hat{t}:=\sum_{n\in N} s_n / |A|$, and that if a committee satisfies $t$-PJR then it also satisfies $t'$-PJR for each $t'\geq t$, i.e., the property gets stronger as $t$ decreases. 
This is in contrast to the maximin support objective, which implies a stronger property as it increases.

We remark that the notion of \emph{average satisfaction} introduced in~\cite{sanchez2017proportional} also attempts to quantify the level of proportional representation achieved by a committee. 
Informally speaking, that notion measures the average number of representatives in the committee that each voter in a group has, for any group of voters with sufficiently high aggregate vote strength and cohesiveness. 
In contrast, with parametric PJR we focus on providing sufficient representatives to the group as a whole and not to each individual voter, and we measure the aggregate vote strength required to gain adequate representation.
Interestingly, the average satisfaction measure is closely linked to the EJR property, and in particular in~\cite{aziz2018complexity} this measure is used to prove that a local search algorithm achieves EJR; 
%similarly, in the full version we use parameteric PJR to prove that a local search version of $\phragmms$ achieves standard PJR.
similarly, in Appendix~\ref{s:LS} we use parametric PJR to prove that a local search version of $\phragmms$ achieves standard PJR.

Testing whether an arbitrary solution satisfies standard PJR is known to be coNP-complete~\cite{aziz2018complexity}, hence the same remains true for its parametric version.
We provide next a sufficient condition for a committee to satisfy $t$-PJR which is efficiently testable, based on our definitions of parameterized score and score.  

\begin{lemma} \label{lem:locality}
If for a feasible solution $(A,w)$ there is a parameter $t\in\R$ such that $\max_{c'\in C\setminus A} \prescore(c',t)<t$, or equivalently, $\max_{c'\in C\setminus A} \score(c') <t$, then committee $A$ satisfies $t$-PJR. 
This condition can be tested in $O(|E|)$ time.
\end{lemma}

\begin{proof} 
We prove the contrapositive of the claim. If $A$ does not satisfy $t$-PJR, there must be a subset $N'\subseteq N$ of voters and an integer $r>0$ that observe properties a) and b) above but fail property c). 
By property a) and the negation of c), set $(\cap_{n\in N'} C_n)\setminus A$ must be non-empty: let $c'$ be a candidate in it. 
We will prove that for any feasible weight vector $w\in \R^E$, it holds that $\prescore(c',t)\geq t$, and consequently $\score(c')\geq t$ by the definition of score. We have
\begin{align*} 
\prescore(c',t) &= \sum_{n\in N_{c'}}  \slack(n,t) 
\geq \sum_{n\in N'} \slack(n,t) \\
&\geq \sum_{n\in N'} \Big(s_n - t\cdot \sum_{c\in A\cap C_n} \frac{w_{nc}}{\supp_w(c)}\Big)  \\
&= \sum_{n\in N'} s_n - t \cdot \sum_{c\in A\cap (\cup_{n\in N'} C_n)} 
\frac{\sum_{n\in N'\cap N_c} w_{nc}}{\sum_{n\in N_c} w_{nc}} \\ 
& \geq t\cdot r - t\cdot |A\cap (\cup_{n\in N'} C_n)| \\
& \geq t\cdot r - t\cdot (r-1) = t, 
\end{align*}
where the first inequality holds as $N'\subseteq N_{c'}$ by our choice of candidate $c'$, the second one holds by definition of slack, the third one holds by property b) and because the fraction on the third line is at most $1$ for each candidate $c$, and the last inequality holds by negation of c). 
This proves that $\prescore(c',t) \geq t$. 

For a given solution $(A,w)$ and parameter $t$, one can verify the condition above in time $O(|E|)$ by computing $\maxprescore(A,w,t)$ and comparing the output to $t$; see Theorem~\ref{thm:runtimes}.
\end{proof}

The proof of Lemma~\ref{lem:localopt} now follows as a corollary.

\begin{proof}[Proof of Lemma~\ref{lem:localopt}]
Let $c_{\max}$ be a candidate with highest score $t_{\max}=\score(c_{\max})=\max_{c'\in C\setminus A} \score(c')$. 
If $\supp_w(A)\geq t_{\max}$, it follows from Lemma~\ref{lem:insert} that if we execute $\ins(A,w,c_{\max}, t_{\max})$, we obtain a solution $(A+c_{\max}, w')$ with $\supp_{w'}(A+c_{\max})=t_{\max}$. 
Now, by feasibility of vector $w'$, we have the inequality $\sum_{n\in N} s_n \geq \sum_{c\in A+c_{\max}} \supp_{w'}(c) \geq (k+1)\cdot t_{\max}$, which implies that $t_{\max}\leq \sum_{n\in N} s_n / (k+1) < \sum_{n\in N} s_n / k =: \hat{t}$. 
By Lemma~\ref{lem:locality} above, having $t_{\max} < \hat{t}$ implies that $A$ satisfies $\hat{t}$-PJR, which is standard PJR.
\end{proof}

We end the section with the observation that inequality~\ref{eq:localopt} corresponds to a notion of local optimality for solution $(A,w)$. 
Indeed, if the inequality did not hold we could improve upon the solution by iteratively swapping the member with least support for the unelected candidate with highest score, resulting in an increase of the least member support and/or a decrease of the highest score among unelected candidates, which by Lemma~\ref{lem:locality} strengthens the level of parametric PJR of the solution. 
Therefore, the fact that $\phragmms$ always returns a locally optimal solution (Lemma~\ref{lem:315localoptimality}) implies standard PJR but constitutes a strictly stronger property. 
%In the full version we formalize this local search algorithm and use it to prove Theorem~\ref{thm:intro3}.
%In Appendix~\ref{s:LS}, we formalize this local search algorithm and use it to prove Theorem~\ref{thm:intro3}.
In the next section, we formalize this local search algorithm and use it to prove Theorem~\ref{thm:intro3}.

%% file: local.tex
\section{A local search algorithm}\label{s:LS}

Suppose that we know of an approximation algorithm for maximin support with no known guarantee to satisfy the PJR property, or we happen to know of a high quality solution in terms of maximin support but we ignore if it satisfies PJR. Can we use it to find a new solution of no lesser quality which also satisfies PJR? And can we efficiently convince a verifier of this fact? We answer these questions in the positive for the first time, and prove Theorem~\ref{thm:intro3}.

We present a local search algorithm that takes an arbitrary full solution as input, and iteratively drops a member of least support and inserts a new candidate with highest score, using the heuristic presented in Section~\ref{s:inserting}. 
The procedure always maintains or increases the value of the least member support, hence the quality of the solution is preserved. Furthermore, as this local search converges to a locally optimal solution, the solutions found along the way are guaranteed to satisfy PJR after a limited number of iterations. 
Therefore, this procedure can be used as an efficient post-computation of any algorithm for maximin support, in a black-box manner, to enable the PJR property.
 This local search variant of $\phragmms$ highlights the robustness of this heuristic; in particular, there is no evident way to build a similar variant from $\phragmen$~\cite{brill2017phragmen}, since its analysis of the PJR property makes assumptions on the structure of the current solution at the beginning of each iteration. 

As we did in Section~\ref{s:inserting}, in the following algorithm we assume that the background instance $(G=(N\cup C, E), s, k)$ is known and that does not need to be passed as input. 
Instead, the input is a feasible full solution $(A,w)$, and a parameter $\eps>0$. 
Our proposed algorithm $\local$ is presented in Algorithm~\ref{alg:localpjr}. 

\begin{algorithm}[htb]
\SetAlgoLined
\KwData{Full feasible solution $(A,w)$, parameter $\eps>0$.}
\tcp{parameter related to def.~of standard PJR} 
Let $\hat{t}\leftarrow \sum_{n\in N} s_n / |A|$ \;
\While{True}{
 \tcp{find member with least support} 
  Find tuple $(c_{\min}, t_{\min})$ so that $c_{\min}\in A$ and $t_{\min}=\supp_w(c_{\min})=\supp_w(A)$ \;
	\tcp{find candidate with highest score} 
  Let $(c_{\max}, t_{\max})\leftarrow \maxscore(A,w)$ \;
  \lIf {($t_{\max}< \min\{ (1+\eps)\cdot t_{\min}, \hat{t}\}$)} {\Return $(A,w)$}
  Update $(A,w)\leftarrow \ins(A-c_{\min},w,c_{\max},t_{\max})$\;
}
\caption{$\LSPJR(A,w,\eps)$}
\label{alg:localpjr}
\end{algorithm}

\begin{theorem}\label{thm:enabler}
For any $\eps>0$ and a feasible full solution $(A,w)$, let $(A',w')$ be the output solution to $\LSPJR(A,w,\eps)$. Then: 
\begin{enumerate}
    \item $(A',w')$ is feasible and full, and $\supp_{w'}(A')\geq \supp_w(A)$; \label{item:support}
		\item if $(A,w)$ has an $\alpha$-approximation guarantee for maximin support, for some $\alpha\geq 1$, then the algorithm performs at most $k\cdot \lfloor 1+\log_{1+\eps} \alpha\rfloor +1$ iterations, each in time $O(|E|\cdot \log k)$; \label{item:iterations}
		\item $(A', w')$ satisfies the condition on Lemma~\ref{lem:locality} for parameter $t=\min\{(1+\eps)\cdot \supp_{w'}(A'), \hat{t}\}$, hence $A'$ verifiably satisfies both standard PJR and $[(1+\eps)\cdot \supp_{w'}(A')]$-PJR; and
		\label{item:tPJR}
		\item by setting $\eps\rightarrow\infty$, the algorithm finds a solution satisfying standard PJR in at most $k+1$ iterations. \label{item:infinity}

\end{enumerate}
\end{theorem}

This proves Theorem~\ref{thm:intro3}. 
Notice that point~\ref{item:iterations} establishes that the algorithm above can be executed as a post-computation of any constant-factor approximation algorithm for maximin support in time $O(|E|\cdot k\log k)$. In particular, this complexity is lower than that of each constant-factor approximation presented in this paper.
By point~\ref{item:infinity}, the algorithm can be sped up if we only care about standard PJR, or can run further iterations to converge to a locally optimal solution and achieve a stronger parametric PJR guarantee. 

\begin{proof}
We start with the easier statements. 
As we update the current solution by calling the $\ins$ algorithm, Lemma~\ref{lem:insert} guarantees that the solution remains feasible and full. 
Next, whenever the algorithm terminates, point~\ref{item:tPJR} follows from Lemma~\ref{lem:locality} and the fact that the stopping condition is satisfied. 
Furthermore, point~\ref{item:infinity} is just a special case of point~\ref{item:iterations}, for a high enough value of $\eps$. 
Hence, it only remains to prove the inequality in point~\ref{item:support}, as well as point~\ref{item:iterations}. 

In what follows we use $i$ as a superscript to indicate the value of the different variables at the beginning of the $i$-th iteration. 
 Let $t^i_{\st}:=\min\{(1+\eps)\cdot t^i_{min}, \hat{t}\}$, so that the stopping condition reads $t^i_{\max} \stackrel{?}{<} t^i_{\st}$. 
Notice that $t^i_{\min}\leq \hat{t}$ always holds by definition of $\hat{t}$ and feasibility of $w^i$, so $t^i_{\min}\leq t^i_{\st}$ holds as well. 
 
We prove the inequality in point~\ref{item:support} by induction on $i$. 
Consider an iteration $i\geq 0$ in which the stopping condition does not hold, so $t^i_{\max} \geq t^i_{\st}\geq t^i_{\min}$. 
On one hand, it is evident that 
$$\supp_{w^i}(A^i-c^i_{\min})\geq \supp_{w^i}(A^i)=t^i_{\min}, $$ 
i.e., the least support can only increase when we drop member $c_{\min}^i$ from the committee. 
On the other hand, 
\begin{align*}
\prescore_{(A^i-c_{\min}^i, w^i)}(c^i_{\max}, t^i_{\max}) &\geq \prescore_{(A^i, w^i)}(c^i_{\max}, t^i_{\max}) 
	 = t^i_{\max}\geq t^i_{\min},
\end{align*}
or more generally, for any fixed candidate, threshold, and edge weight vector, a parameterized score can only increase when a committee member is dropped. Such fact follows from the definitions of slack and parameterized score. 
Thus, by Lemma~\ref{lem:insert}, 
\begin{align*}
\supp_{w^{i+1}} (A^{i+1}) \geq \min & \Big\{ \supp_{w^i}(A^i-c^i_{\min}), t^i_{\max}, 
	 \prescore_{(A^i-c_{\min}^i, w^i)}(c^i_{\max}, t^i_{\max})\Big\} \geq t^i_{\min}, 
\end{align*}
which is what we needed to show.

We continue to point \ref{item:iterations}. The complexity of an iteration is dominated by the call to $\maxscore$, which takes time $O(|E|\cdot \log k)$ by Theorem~\ref{thm:runtimes}. 
It remains to prove the bound on the number $T$ of total iterations. 
To do that, we analyze the evolution of the least member support $t^i_{\min}=\supp_{w^i}(A^i)$. 
\emph{Claim A:} If ever $t^i_{\min}=\hat{t}$, then the algorithm terminates immediately, i.e., $i=T$. This is because in this case all members in $A^i$ must have a support of exactly $\hat{t}$, all voters a zero slack for threshold $\hat{t}$, and all candidates in $C\setminus A^i$ a zero parameterized score for $\hat{t}$, and hence a score strictly below $\hat{t}$, and the stopping condition is fulfilled.
\emph{Claim B:} For any iteration $i$ with $1\leq i<T-k$, we have $t^{i+k}_{\min}\geq (1+\eps)\cdot t^{i}_{\min}$. This is because, by Lemma~\ref{lem:insert}, in each iteration $j\geq i$ we are removing a member of least support while not increasing the number of members with support below 
$$\prescore_{(A^j - c^j_{\min}, w^j)}(c^j_{\max}, t^j_{\max})\geq t^j_{\max}\geq t^j_{\st}\geq t^i_{\st},$$
where we used the obvious fact that $t^i_{\st}$ is monotone increasing throughout the iterations.  
As there are only $k$ committee members, it takes at most $k$ iterations to remove all members with support below $t^i_{\st}$, so we must have that $t^{i+k}_{\min}\geq t^i_{\st}=\min\{(1+\eps)\cdot t^i_{\min}, \hat{t}\}$. 
Since the $(i+k)$-th iteration is not the last one, and by Claim A,  $t^{i+k}$ cannot be higher than $\hat{t}$, so it must be higher than $(1+\eps)\cdot t^i_{\min}$, which proves Claim B.
Therefore, if the algorithm outputs a solution $(A^T,w^T)$ and has an input solution $(A^1, w^1)$ with an $\alpha$-approximation guarantee for maximin support, then 
$$\alpha\cdot \supp_{w^1}(A^1) \geq \supp_{w^T}(A^T) \geq (1+\eps)^{\lfloor (T-1)/k \rfloor} \cdot \supp_{w^1}(A^1).$$
Hence, $\alpha\geq (1+\eps)^{\lfloor (T-1)/k \rfloor}$, and $T\leq k\cdot \lfloor 1+\log_{1+\eps} \alpha \rfloor+1$. This completes the proof of point \ref{item:iterations}.
\end{proof}

%% file: implementation.tex
\section{Validator selection for NPoS}\label{s:implement}

Recall that our motivating application is the selection of validators for a blockchain network that implements Nominated Proof-of-Stake (NPoS). 
In this section we provide further details about the NPoS mechanism, and sketch a proposal for an implementation of a validator selection protocol that uses the new election rule. We consider the Polkadot network as a specific example.

We have chosen to use approval-based voting because, as argued in~\cite{laslier2010handbook}, this electoral system is easy to understand for voters -- no need to rank or grade candidates and no limit on the number of votes -- while simultaneously benefiting from a rich literature -- which we have obviously exploited to derive theoretical guarantees for our protocol. 
We also mention the importance of allowing nominators to vote for more than one candidate: otherwise, they face the dilemma of having to choose between a popular validator and a validator that may represent them better but has a lower chance of being elected. A rational nominator seeking to maximize staking rewards will then prefer to vote for the popular validator, and this type of tactical voting will result in popular validators gaining more and more stake backing over time, to the detriment of the network's decentralization goal. 

We highlight the utility of having the validator election protocol output a stake distribution vector $w\in \R^E$ along with the winning committee $A$, beyond its role as a witness in the verification of guarantees. 
This vector effectively defines $k$ disjoint staking pools, of size $\sup_w(c)$ for each validator $c\in A$, to which nominators belong fractionally. 
These pools are created automatically by the election rule, with a distribution of sizes as close to uniform as possible. And they change dynamically from one election to the next, so that when a validator leaves the system its corresponding pool is automatically disintegrated and the affected nominators reassigned to other pools in a balanced way.

In terms of incentives, pools are paid in proportion to the validators' performance, but \emph{independently of the pool sizes}, while within each pool the nominators' rewards are proportional to their stakes. Hence, nominators belonging to smaller pools tend to get paid more per token unit. %
%This mechanism has two interesting consequences. First, 
As such, if the election rule outputs a balanced solution then it automatically maximizes the rewards of all nominators. 
More precisely, if we assume that the winning committee is fixed and all pools are paid the same in expectation, and consider a staking mechanism that allows each nominator to freely distribute her stake among the pools of the validators she trusts, it can be checked that \emph{a vector of stake distributions is a Nash equilibrium precisely when it is balanced}. 
In this sense, there is little motivation to vote strategically under NPoS. 
%The second consequence is that nominators are economically encouraged to shift their preferences over time towards less popular validators. As such, we can expect the distribution of stake pool sizes not to ever deviate greatly from uniform, which benefits the security goal of the protocol.

As of late May 2021, Polkadot selects a committee of $k\approx 300$ validators to participate in its consensus protocol, out of a set of $|C|\approx 900$ candidates, and has $|N|\approx 20000$ nominators. 
It runs a validator selection protocol once per era, i.e., roughly once per day. 
Towards the end of each era, the protocol enters an \emph{election phase}, where it takes a snapshot of the current nominators' preferences and stakes, and uses it as input to run a committee election rule to select the committee of $k$ validators for the following era. 

We propose running the $\phragmms$ rule as a verifiable computing scheme. 
The current validators -- or a large enough subset thereof -- may act themselves as off-chain workers and run $\phragmms$ as a background task logically separate from consensus and with a relaxed time frame of (say) up to an hour. 
Whichever off-chain worker finishes the task first can submit the (prospective) solution on-chain as a transaction, over which the protocol will execute the verification process described in Theorem~\ref{thm:315guarantee}.

Experimental results show that for the current number of nominators, the verification can be run in a single block. Hence, by applying the parallelization presented in Lemma~\ref{lem:parallel} and distributing this computation over (say) $p=10$ consecutive blocks, the system will be able to handle over 200000 nominators. As Polkadot produces roughly ten blocks per minute, such a solution would be verified within a minute.

Optionally, the system may also allow any user to submit a prospective solution as a transaction during the election phase: this way, it benefits from potentially better committees found with other election rules. 
To this end, notice that any user running a different election rule can also run the post-computation described in Theorem~\ref{thm:intro3}, in order to satisfy PJR verifiably. Then, validators can accept this solution only if it passes the corresponding PJR test on-chain and its maximin support objective is better than that of $\phragmms$.  
Yet, anti-spam measures may be required in this case, such as a very high transaction fee.

We claim that being able to verify that the winning solution provides strong guarantees on security and proportionality -- as opposed to, e.g., simply selecting among all submitted solutions the one with highest maximin support objective --  protects the network against a possible long-range attack, as we explain now. 
Consider a scenario where an adversary currently controlling a minority of validators creates a private fork (i.e., an alternative valid chain) right before the start of the election phase, and in this fork he censors all prospective solutions except for one, fabricated by the adversary himself, in which he becomes grossly overrepresented and possibly even gains control of a majority of validators. 
Then, after the end of the election phase he publishes this fork and attempts to make it canonical (for instance by making it longer than all other forks, if the consensus protocol follows the longest-chain rule). 
If this attack is successful, the adversary will have captured the network in the next era. 
To resist such an attack, we insist on discarding all submitted solutions that do not pass the verification test, and propose that the election phase may extend indefinitely until at least one solution passing the test has been submitted. 
%It is worth mentioning that Polkadot is also protected against long-range attacks by its use of a finality gadget~\cite{stewart2020grandpa}, which limits the possibility of long chain reorganizations.

%We just check feasibility.
%Check PJR test offchain.
%Proof of not PJR (PJR challenge): signal one candidate with high score, also give its supporters (could have hundreds). Computing slacks only requires output graph.
%Easy to prove that solution is not PJR, and polytime to find counterexample. 

%Right now we're running Phragmen, with random number of star balancing rounds. This is what validators should run.
%Solution is checked off-chain, solution is checked on chain for feasibility.
%Validators fight over submitting better solutions, they need to be epsilon better and then they overwrite each other's solution.

%Eventually the idea is to run Phragmms, but only once we have solutions from community + bot from W3F.
%Chain doesn't care what algorithm is run off chain.

%Idea: we need absolute notion of quality of committee
%What if there's governance-controlled least support? Solution is rejected if this check doesn't pass.
%That would be a good guarantee.
%What if: lots of stake needs to be used.
%Much safer! At least as temporary patch.

%% file: conclusion.tex
\section{Conclusions and open questions}\label{s:conc}

There is a recent surge of proof-of-stake (PoS) based blockchain projects that take a ``representative democracy'' approach, where token holders choose the validator candidates that they trust, and out of this process a committee with $k$ of the most trusted candidates emerges as the set of active validators. 
Although these \emph{validator selection protocols} play a critical role in preserving the security and decentralization levels of their networks, we observe a multitude of design choices with scarce formal justifications behind them. 
We present the first computational social choice analysis for the electoral system of a validator selection protocol, namely for Nominated Proof-of-Stake (NPoS). 
Starting from first principles and in the pursuit of security and decentralization, we formalize the problem in terms of proportional justified representation (PJR) and maximin support, two criteria recently introduced in the literature of proportional representation.
With the problem definition at hand, we show that current election rules either perform poorly in terms of security, or are too slow to be compatible with the blockchain architecture. We then propose $\phragmms$, the first rule to provide formal guarantees on both criteria and be implementable on a blockchain network. 

We propose the adaptation of committee election rules as verifiable computing schemes. Indeed, this adaptation proves to be key for the implementation of the $\phragmms$ rule in a blockchain architecture. 
We remark that our proposed verification process (Theorem~\ref{thm:315guarantee}) is linear in the size of the input, and leave it as an open challenge to find an election rule that achieves a constant-factor approximation guarantee for maximin support verifiably, whose verification process has a runtime \emph{sublinear in the number of voters}. For instance, succinct non-interactive arguments of knowledge (SNARKs) may be of relevance.

Contrary to the common understanding of proportional representation as a criterion that combats underrepresentation, we formalize the goal of preventing the overrepresentation of any faction within the electorate, and present new results for it. These results may be applicable beyond distributed networks, whenever a governance body needs to maintain a delicate balance of power among groups in contention, and may be at particular risk of capture or disruption if this balance is lost. 

We present the first approximability analysis for a Phragm\'{e}n objective: we show that the maximin support objective can be approximated within a factor of 2, but not within a factor of $1.2-\eps$ for any $\eps>0$ unless P=NP. This gap between approximability and hardness awaits to be closed in future work. 
Similarly, it remains open to establish whether the approximation guarantee factors we proved for $\MMS$ and $\phragmms$ ($2$ and $3.15$ respectively) are tight.

We highlight that our approximation analyses are based on network flow theory, a promising tool not widely used in the literature of committee election rules. 
Similarly, we leverage Phragm\'{e}n's notion of load balancing, formalize what it means for a vote distribution over the edges of the approval graph to be balanced, and show how to find a balanced distribution efficiently using new results related to parametric flow. We then synthesize the heuristics behind $\phragmen$, $\MMS$ and $\phragmms$ in terms of how well partial solutions are rebalanced in between iterations. 
Further election rules might be analyzed in the future using network flow theory and our definition of balancedness.

%% file: balanced.tex
\section{Computing a balanced solution} \label{s:balanced}

Recall from Section~\ref{s:prel} that for a background election instance $(G = (N \cup C, E), s, k)$ and a fixed committee $A\subseteq C$, a weight vector $w\in \R^E$ is balanced for $A$ if a) it maximizes the sum of member supports, $\sum_{c\in A} supp_w(c)$, over all feasible weight vectors, and b) it minimizes the sum of supports squared, $\sum_{c\in A} (supp_w(c))^2$, over all vectors that observe the previous property. 
In this section we provide algorithms to compute such a vector.

We start by noticing that a balanced weight vector can be computed with numerical methods for quadratic convex programs. 
Let $E_A\subseteq E$ be the restriction of the input edge set $E$ over edges incident to committee $A$, and let $D\in\{0,1\}^{A\times E_A}$ be the vertex-edge incidence matrix for $A$. 
For any weight vector $w\in\R^{E_A}$, the support that $w$ assigns to candidates in $A$ is given by vector $Dw$, so that $supp_w(c)=(Dw)_c$ for each $c\in A$. 
We can now write the problem of finding a balanced weight vector as a convex program:
\begin{align*}
    \text{Minimize} \quad & \|Dw\|^2 \\
    \text{Subject to } \quad & w\in\R^{E_A}, \\
    & \sum_{c\in C_n} w_{nc} \leq s_n \quad \text{for each } n\in N, \text{ and} \\
    & \mathbbm{1}^{\intercal} Dw = \sum_{n\in \cup_{c\in A} N_c} s_n,
\end{align*}
where the first line of constraints corresponds to non-negativity, the second one to feasibility (see inequality~\ref{eq:feasible}), and the last line ensures that the sum of supports is maximized (see property 2 in Lemma~\ref{lem:balanced}), where $\mathbbm{1}\in\mathbb{R}^A$ is the all-ones vector.

However, there is a more efficient method using techniques for parametric flow, which we sketch now. Hochbaum and Hong~\cite[Section 6]{hochbaum1995strongly} consider a network resource allocation problem which generalizes the problem of finding a balanced weight vector: given a network with a single source, single sink and edge capacities, minimize the sum of squared flows over the edges reaching the sink, over all maximum flows. 
They show that this is equivalent to a parametric flow problem called \emph{lexicographically optimal flow}, studied by Gallo, Gregoriadis and Tarjan~\cite{gallo1989fast}. 
In turn, in this last paper the authors show that, although a parametric flow problem usually requires solving several consecutive max-flow instances, this particular problem can be solved running a single execution of the FIFO preflow-push algorithm by Goldberg and Tarjan~\cite{goldberg1988new}.

Therefore, the complexity of finding a balanced weight vector is bounded by that of Goldberg and Tarjan’s algorithm, which is $O(n^3)$ for a general $n$-vertex network. 
However, Ahuja et al.~\cite{ahuja1994improved} showed how to optimize several popular network flow algorithms for the case of bipartite networks, where one of the partitions is considerably smaller than the other. Assuming the sizes of the bipartition are $n_1$ and $n_2$ with $n_1 \ll n_2$, they implement a two-edge push rule that allows one to "charge" most of the computation weight to the vertices on the small partition, and hence obtain algorithms whose running times depend on $n_1$ rather than $n$. 
In particular, they show how to adapt Goldberg and Tarjan’s algorithm to run in time $O(e\cdot n_1+n_1^3)$, where $e$ is the number of edges. 
For our particular problem, which can be defined on a bipartite graph $(N\cup A, E_A)$ where $|A|\leq k\ll |N|$, we obtain thus an algorithm that runs in time $O(|E_A|\cdot k + k^3)$.

%% file: flow.tex
\section{Network flow results}\label{s:flow}

In the section we present proofs to Lemmas \ref{lem:2sols} and \ref{lem:N_a}, the key results at the heart of the constant-factor approximation guarantees we derive for the $\MMS$ and the $\phragmms$ rules, respectively. 
We start by introducing some necessary definitions and results related to network flow theory.

Throughout this section we regard the input bipartite approval graph $G=(N\cup C,E)$ as a network, and regard an edge vector $f\in\mathbb{R}^{E}$ as a \emph{flow} over it. For each edge $nc\in E$, we consider the flow on that edge to be directed toward $c$ if $f_{nc}>0$, and directed toward $n$ if $f_{nc}<0$. 
Consequently, the \emph{excess} of a voter $n\in N$ is defined as $e_f(n):=\sum_{c\in C_n} f_{nc}$, and the excess of a candidate $c\in C$ is defined as $e_f(c):=-\sum_{n\in N_c} f_{nc}$. 
%A vertex $x\in N\cup C$ is an \emph{excess} vertex if $e_f(x)>0$, and a \emph{deficit} vertex if $e_f(x)<0$. 
Moreover, for a set of vertices $S\subseteq N\cup C$, we define its \emph{net excess} as $e_f(S):=\sum_{x\in S} e_f(x)$.
A vector $f'\in\mathbb{R}^E$ is a \emph{sub-flow of $f$} if 
\begin{itemize}
	\item for each edge $nc\in E$ with $f'_{nc}\neq 0$, flows $f'_{nc}$ and $f_{nc}$ have equal sign (i.e., direction) and $|f'_{nc}|\leq |f_{nc}|$, and
	\item for each vertex $x\in N\cup C$ with $e_{f'}(x)\neq 0$, excesses $e_{f'}(x)$ and $e_{f}(x)$ have equal sign and $|e_{f'}(x)|\leq |e_{f}(x)|$.
\end{itemize}

We now list two properties of flows and sub-flows. 
The proof of Theorem~\ref{thm:decomposition} can be found in \cite[Thm.~3.15]{ahuja1994improved}, while the proof of Lemma~\ref{lem:subflow} is delayed to Appendix~\ref{s:proofs}.

\begin{theorem}[Flow Decomposition Theorem]\label{thm:decomposition}
Any flow $f\in\mathbb{R}^E$ can be decomposed into a finite number of cycles and simple paths, such that every path $p$ is a non-zero sub-flow of $f$ that starts in a vertex with strictly positive excess and ends in a vertex with strictly negative excess. 
\end{theorem}

\begin{lemma}\label{lem:subflow}
If $w, w'\in\R^E$ are two non-negative and feasible edge weight vectors for the given election instance, and $f'\in\mathbb{R}^E$ is a sub-flow of $f:=w'-w$, then both $w+f'$ and $w'-f'$ are non-negative and feasible as well.
\end{lemma}

We prove now that for any partial solution, there is always an unelected candidate that can be added with large support. To show this, we assume we know the edge weight vector of an optimal solution, and combine it with the weight vector of the current solution. For convenience, we repeat the statement of Lemma~\ref{lem:2sols}.

\begin{lemma*}
If $(A^*, w^*)$ is an optimal solution to maximin support, and $(A,w)$ is a partial solution with $|A|\leq k$ and $A\neq A^*$, there is a candidate $c'\in A^*\setminus A$ and feasible solution $(A+c', w')$ such that 
$$\supp_{w'}(A+c')\geq \min\Big\{\supp_w(A), \frac{1}{2} \supp_{w^*}(A^*)\Big\}.$$
\end{lemma*}

\begin{proof}
Let $(A,w)$ and $(A^*, w^*)$ be as in the statement, with corresponding values $t^*:=\supp_{w^*}(A^*)$ and $t:=\min\{\supp_w(A), t^*/2\}$. To prove the lemma, it suffices to find a candidate $c'\in A^*\setminus A$ and a feasible weight vector $w'\in\R^E$ such that $\supp_{w'}(A+c)\geq t$.

By decreasing some components in $w$ and $w^*$, we can assume without loss of generality that $\supp_w(c)=t$ if $c\in A$, zero otherwise, and $\supp_{w^*}(c)=t^*$ if $c\in A^*$, zero otherwise. 
Consider flow $f:=w^* - w\in\mathbb{R}^E$ over the network induced by $N\cup A\cup A^*$. 
%We partition this network into four subsets: $N$, $A\setminus A^*$, $A^*\setminus A$ and $A^*\cap A$. 
It is easy to see that 
\begin{itemize}
\item $N$ has a net excess $e_f(N)=|A^*|\cdot t^* - |A|\cdot t$,
\item $A\setminus A^*$ has a net excess $e_f(A\setminus A^*)=|A\setminus A^*|\cdot t$,
\item $A^*\setminus A$ has a net excess $e_f(A^*\setminus A)=-|A^*\setminus A|\cdot t^*$, and
\item $A^*\cap A$ has a net excess $e_f(A^*\cap A)=-|A^*\cap A| \cdot (t^*-t)$.
\end{itemize}

By Theorem~\ref{thm:decomposition}, we can decompose flow $f$ into circulations and simple paths, where each path is a sub-flow of $f$ that starts (ends) in a vertex with positive (negative) excess. 
Let $f'$ is the sum of all paths that start inside subset $N\cup (A^*\cap A)$ and end outside of it. 
It follows that $f'$ is also a sub-flow of $f$, and that the amount of flow it extracts from this subset is at least
\begin{align*}
e_f(N\cup(A^*\cap A)) &= e_f(N) + e_f(A^*\cap A)\\
&= |A^*|\cdot t^* - |A|\cdot t - |A^*\cap A| \cdot (t^*-t)\\
&= |A^*\setminus A|\cdot t^* - |A\setminus A^*|\cdot t \\
&\geq |A^*\setminus A|\cdot (t^*-t) \geq |A^*\setminus A|\cdot t,
\end{align*}
where the last two inequalities follow from $|A^*|\geq |A|$ and $t\leq t^*/2$, respectively.

Next, we claim that each path in $f'$ actually must start in $N$ and end in $A^*\setminus A$. 
Indeed, none of these paths can start in $A^*\cap A$, because each vertex in that set has negative excess, and similarly none can end in $A\setminus A^*$, because each vertex in that set has positive excess. 
Therefore, $f'$ carries flow exclusively from $N$ to $A^*\setminus A$, and by the previous inequality and an averaging argument, there must be a vertex $c'$ in $A^*\setminus A$ that receives a flow from $f'$ of value at least $t$. 

Finally, if we define vector $w':=w+f'$, it is non-negative and feasible by Lemma~\ref{lem:subflow}, and it provides the same supports to the members of $A$ as $w$ does, namely $t$, and a support of at least $t$ to candidate $c'$. Hence, $\supp_{w'}(A+c')\geq t$, as claimed.   
\end{proof}

Next, we prove that if we start with a partial solution that is balanced, then not only are there unelected candidates that can be appended with high support, but they also have large scores, so we can find such a candidate efficiently with our heuristic. 
To show this, we prove that there must be a subset of voters with large aggregate vote strength who do not get as many representatives in the current solution as they do in the optimal solution, so they have large slacks and their unelected representatives have large scores. For convenience, we repeat the statement of Lemma~\ref{lem:N_a}.

\begin{lemma*}
If $(A^*, w^*)$ is an optimal solution with $t^*=\supp_{w^*}(A^*)$, and $(A,w)$ is a balanced solution with $|A|\leq k$ and $A\neq A^*$, then for each $0\leq a\leq 1$ there is a subset $N(a)\subseteq N$ of voters such that 
\begin{enumerate}
	\item each voter $n\in N(a)$ approves of a candidate in $A^*\setminus A$;
	\item for each voter $n\in N(a)$, we have $\supp_w(A\cap C_n)\geq at^*$;
	\item $\sum_{n\in N(a)} s_n \geq |A^* \setminus A|\cdot (1-a) t^*$; and
	\item for any $b$ with $a\leq b\leq 1$ we have that $N(b)\subseteq N(a)$, and a voter $n\in N(a)$ belongs to $N(b)$ if and only if $n$ observes property 2 above with parameter $a$ replaced by $b$.
\end{enumerate}
\end{lemma*}

\begin{proof}
Fix a parameter $a$, and partition committee $A$ into two sets by defining $A_{hi}:=\{c\in A: \ \supp_w(c)\geq at\}$ and $A_{lo}:=A\setminus A_{hi}$. 
Similarly, partition $N$ into two sets by defining $N_{hi}:=\{n\in N: \ C_n \cap A_{lo}=\emptyset\}$ and $N_{lo}:= N\setminus N_{hi}$. 
If we define $N(a)\subseteq N_{hi}$ as those voters in $N_{hi}$ that have adjacent candidates in $A^*\setminus A$, then properties 1, 2 and 4 become evident. Hence, it only remains to prove the third property.

\emph{Claim A.} There is no edge with non-zero weight in $w$ between $N_{lo}$ and $A_{hi}$. 
Indeed, if there was a pair $n\in N_{lo}$, $c\in A_{hi}\cap C_n$ with $w_{nc}>0$, then by property 3 of Lemma~\ref{lem:balanced} we would have $\supp_w(A\cap C_n)=\supp_w(c)\geq at^*$, contradicting the fact that $n$ is not in set $N_{hi}$. 
Thus, all of the vote strength from voters in $N_{lo}$ must be directed to members in $A_{lo}$, and we get the inequality
$$\sum_{n\in N_{lo}} s_n \leq \sum_{c\in A_{lo}} \supp_w(c) < |A_{lo}|\cdot at^*.$$

Next, by decreasing some components in vectors $w$ and $w^*$, we can assume without loss of generality that $\supp_{w^*}(c)=t^*$ if $c\in A^*$, zero otherwise, and $\supp_{w}(c)=a t^*$ if $c\in A^*\cap A_{hi}$, zero otherwise; we call this the ``wlog assumption'', and notice that $w$ and $w^*$ are not necessarily balanced anymore, but are still feasible.
Consider flow $f:=w^* - w\in\mathbb{R}^E$ over the network induced by $N\cup A^*$. 
We partition the vertices of this network into five subsets: $N_{hi}$, $N_{lo}$, $A^*\setminus A$, $A^* \cap A_{hi}$ and $A^*\cap A_{lo}$; see Figure~\ref{fig:sets}.  
It is easy to see that 
\begin{itemize}
\item $A^*\cap A_{hi}$ has a net excess $e_f(A^*\cap A_{hi})=-|A^*\cap A_{hi}|\cdot (1-a)t^*$,
\item $N$ has a net excess $e_f(N)=|A^*|\cdot t^* - |A^*\cap A_{hi}|\cdot a t^*$, and
\item $N_{lo}$ has a net excess $e_f(N_{lo})\leq \sum_{n\in N_{lo}} s_n < |A_{lo}|\cdot at^*$.
\end{itemize} 

\begin{figure}[htb]
  \centering
	\includegraphics[width=0.55\linewidth]{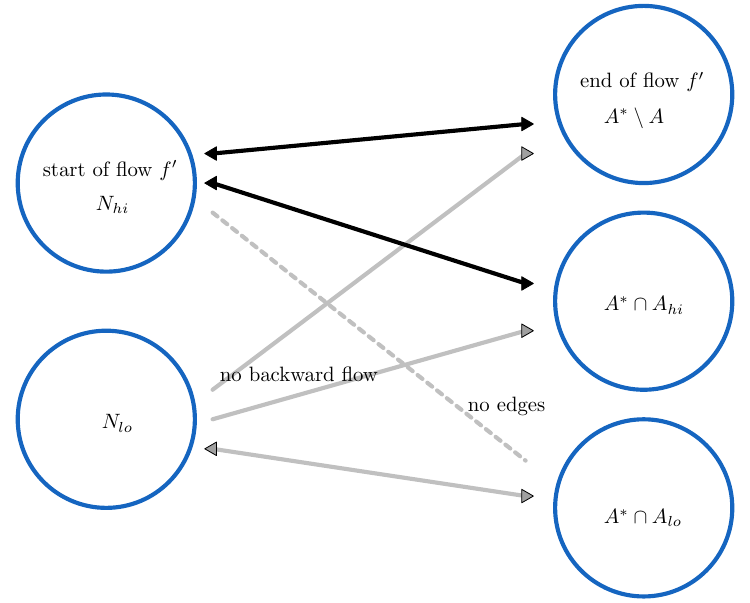}
  \caption{Flow $f'$ starts in $N_{hi}$ and must end in $A^*\setminus A$, because it cannot visit $N_{lo}$ nor $A^*\cap A_{lo}$.}
  \label{fig:sets}
\end{figure}

By Theorem~\ref{thm:decomposition}, we can decompose flow $f$ into circulations and simple paths, where each path is a sub-flow of $f$. 
Let $f'$ be the sum of all paths that start inside subset $N_{hi}\cup(A^*\cap A_{hi})$ and end outside of it. 
It follows that $f'$ is also a sub-flow of $f$, and the amount of flow it extracts from this subset is at least
\begin{align*}
e_f  (N_{hi}\cup(A^*\cap A_{hi})) 
=& e_f(N) - e_f(N_{lo}) + e_f(A^*\cap A_{hi}) \\
>& |A^*|\cdot t^* - |A^*\cap A_{hi}|\cdot a t^* - |A_{lo}|\cdot at^* - |A^*\cap A_{hi}|\cdot (1-a)t^* \\
=& |A^*\setminus A_{hi}|\cdot t^* - |A_{lo}|\cdot at^* \\
\geq & |A^* \setminus A_{hi}|\cdot (1-a)t^* \geq |A^* \setminus A|\cdot (1-a)t^*.
\end{align*}
Now, where does all this flow go?

\emph{Claim B.} Every path in $f'$ must start in $N_{hi}$ and end in $A^*\setminus A$. 
Indeed, it must start in $N_{hi}$ because each vertex in $A*\cap A_{hi}$ has negative excess. 
Furthermore, there are no edges between $N_{hi}$ and $A^*\cap A_{lo}$ (by definition of $N_{hi}$), and there is no flow possible from $(A^*\cap A_{hi})\cup (A^*\setminus A)$ to $N_{lo}$ in $f=w^*-w$ because $w$ has no flow from $N_{lo}$ toward $A_{hi}$ (by Claim A) nor toward $A^*\setminus A$ (by the wlog assumption); see Figure~\ref{fig:sets}. 
This shows that paths in $f'$ cannot visit any vertex in $N_{lo}$ or $A^*\cap A_{lo}$, so they must end in $A^*\setminus A$.

Therefore, paths in $f'$ carry a flow of at least $|A^* \setminus A|\cdot (1-a)t^*$ toward $A^*\setminus A$. 
Finally, for each such path, the last edge goes from $N_{hi}$ to $A^*\setminus A$, so it originates in $N(a)$. 
This proves that $\sum_{n\in N(a)} s_n> |A^* \setminus A|\cdot (1-a) t^*$, which is the third property.
\end{proof}

%% file: algorithms.tex
\section{Algorithmic considerations for the new election rule}\label{s:algorithms}

The goal of this section is threefold. 
First, we prove Theorem~\ref{thm:runtimes} and establish how our heuristic for candidate selection, described in Section~\ref{s:inserting}, can be computed efficiently. 
Second, we improve upon the runtime analysis of $\phragmms$ given in Section~\ref{s:315}, and show that each iteration can be executed in time $O(\bal + |E|)$, down from $O(\bal + |E|\cdot \log k)$. 
Finally, we provide further details on the similarities and differences between $\phragmms$ and $\phragmen$.

As we did in Section~\ref{s:inserting}, we assume in the following that the election instance $(G=(N\cup C, E), s, k)$ is known and does not need to be given as input. Instead, the input is a partial solution $(A,w)$ with $|A|\leq k$. The list of member supports $(supp_w(c))_{c\in A}$ is implicitly passed by reference and updated in every algorithm.

We start with Algorithm~\ref{alg:maxprescore}, which shows how to find the candidate with highest parameterized score for a given threshold $t$.

\begin{algorithm}[htb]
\SetAlgoLined
\KwData{Partial solution $(A,w)$, threshold $t\geq 0$.}
\lFor{each voter $n\in N$}{
compute $slack(n,t)=s_n-\sum_{c\in A\cap C_n} w_{nc}\cdot \min\{1, t/supp_w(c)\}$
}
\lFor{each candidate $c'\in C\setminus A$}{
compute $\prescore(c',t)=\sum_{n\in N_{c'}} \slack(n,t)$
}
Find a candidate $c_t\in\argmax_{c'\in C\setminus A} \prescore(c', t)$\;
\Return $(c_t, \prescore(c_t, t))$\;
 \caption{$\maxprescore(A,w,t)$}
\label{alg:maxprescore}
\end{algorithm}

\begin{lemma}
For a partial solution $(A,w)$ and threshold $t\geq 0$, $\maxprescore(A,w,t)$ executes in time $O(|E|)$ 
and returns a tuple $(c_t,p_t)$ such that $c_t\in C\setminus A$ 
and 
$$p_t=\prescore(c_t,t)=\max_{c'\in C\setminus A} \prescore(c',t).$$
\end{lemma}

\begin{proof}
The correctness of the algorithm directly follows from the definitions of slack and parameterized score. The running time is $O(|E|)$ because each edge in the approval graph $G=(N\cup V, E)$ is inspected at most once in each of the two loops. The first loop also inspects each voter, but we have $|N|=O(|E|)$ since we assume that $G$ has no isolated vertices.
\end{proof}

We move on to computing the highest score. 
For a fixed partial solution $(A,w)$ and for a candidate $c'\in C\setminus A$, consider the function 
\begin{align}\label{eq:scorefunction}
f_{c'}(t):=\prescore(c',t)-t
\end{align}
$$$$ 
in the interval $[0,\infty)$. 
Notice from the definition of parameterized score that this function is convex, continuous and strictly decreasing with no lower bound, and that $f_{c'}(0)\geq 0$; hence it has a unique root corresponding to $score(c')$. We could approximate this root via binary search -- however, we can do better. 
Function $f_{c'}(t)$ is piece-wise linear: if we sort the member supports $\{supp_w(c): \ c\in A\}=\{t_1, \cdots, t_r\}$ so that $t_1 < \cdots < t_r$ for some $r\leq |A|$, then $f_{c'}(t)$ is linear in each interval $[0, t_1), [t_1, t_2), \cdots, [t_r, \infty)$.
Similarly, 
$$f_{\max}(t):= \max_{c'\in C\setminus A} f_{c'}(t) = \max_{c'\in C\setminus A} \prescore(c',t) -t$$ 
is a continuous and strictly decreasing function in the interval $[0,\infty)$, with a unique root $t_{\max}=\max_{c'\in C\setminus A} score(c')$. Unfortunately, this function is in general not linear within each of the intervals above.%
\footnote{It is easy to see that function $f(t)$ is piece-wise linear with $O(|C|\cdot k)$ pieces in total. Hence, one could find its root via binary search by making $O(\log |C|+ \log k)$ calls to $\maxprescore$. 
We present a better approach that only requires $O(\log k)$ such calls.} %
Still, it will be convenient to use binary search to identify the interval that contains $t_{\max}$. We do so in Algorithm~\ref{alg:interval}. The next lemma follows from our exposition and its proof is skipped.

\begin{algorithm}[htb]
\SetAlgoLined
\KwData{Partial solution $(A,w)$.}
Sort the member supports to obtain $0=t_0<t_1<\cdots <t_r$, where $\{t_1, \cdots, t_r\}=\{supp_w(c): \ c\in A\}$\;
\lIf{$p_{t_r}\geq t_r$ where $(c_{t_r},p_{t_r})\leftarrow \maxprescore(A,w,t_r)$}{
	\Return $t_r$}
Let $j_{lo}=0$, $j_{hi}=r-1$\;
\While{$j_{lo}<j_{hi}$}{
  Let $j=\lceil (j_{lo}+j_{hi})/2 \rceil$\;
  \leIf{$p_{t_j}\geq t_j$ where $(c_{t_j},p_{t_j})\leftarrow \maxprescore(A,w,t_j)$}{
  Set $j_{lo}\leftarrow j$}{
  Set $j_{hi}\leftarrow j-1$}
}
\Return $t_{j_{lo}}$\;

 \caption{$\interval(A,w)$}
\label{alg:interval}
\end{algorithm}

\begin{lemma}\label{lem:interval}
For a partial solution $(A,w)$, $\interval(A,w)$ makes $O(\log |A|)$ calls to $\maxprescore$, and thus runs in time $O(|E|\cdot \log k)$. It returns a value $t'$ with $t'\leq t_{\max}:=\max_{c'\in C\setminus A} \score(c')$, and such that for each candidate $c'\in C\setminus A$, the value of $\prescore(c',t)$ is linear in $t$ within the interval $[t',t_{\max}]$.
\end{lemma}

Moving on, for a candidate $c'\in C\setminus A$ and a value $x\geq 0$, consider the linearization of function $f_{c'}(t)$ at $x$ -- more precisely, the linear function that coincides with $f_{c'}(t)$ over the interval $[x, x+\eps]$ as $\eps>0$ tends to zero. 
If we denote by $r_{c', x}$ the unique root of this linearization, we have that
\begin{align*}
    0=& f_{c'}(r_{c', x})|_{\text{linearized at } x}\\
    =& \prescore(c', r_{c', x})|_{\text{linearized at } x} - r_{c', x}\\
    =& \sum_{n\in N_{c'}} \slack(n,r_{c', x})|_{\text{linearized at } x} - r_{c', x}\\
    =& \sum_{n\in N_{c'}} \bigg( s_n - \sum_{c\in A\cap C_n: \ \supp_{w}(c)< x}w_{nc} 
		- \sum_{c\in A\cap C_n: \ \supp_w(c)\geq x} \frac{w_{nc}\cdot r_{c', x} }{\supp_w(c)} \bigg) - r_{c', x},
\end{align*} 
where we used the definitions of parameterized score and slack. Solving for $r_{c', x}$, we obtain
\begin{align}\label{eq:linearized}
    r_{c', x}=\frac{\sum_{n\in N_{c'}} \Big( s_n - \sum_{c\in A\cap C_n: \ \supp_w(c)< x} w_{nc} \Big)}%
    {1+\sum_{n\in N_{c'}} \sum_{c\in A\cap C_n: \ \supp_w(c)\geq x} \frac{w_{nc}}{\supp_w(c)}}.
\end{align}

We make a couple of remarks about these linearization roots. 
First, since $f_{c'}(t)$ is a convex decreasing function, any linearization will lie to its left, and in particular any linearization root will lie to the left of its own root, i.e., 
$$r_{c', x}\leq \score(c') \quad \text{for each } c'\in C\setminus A \text{ and each } x\geq 0.$$

\begin{figure}[htb]
  \centering
	\includegraphics[width=0.6\linewidth]{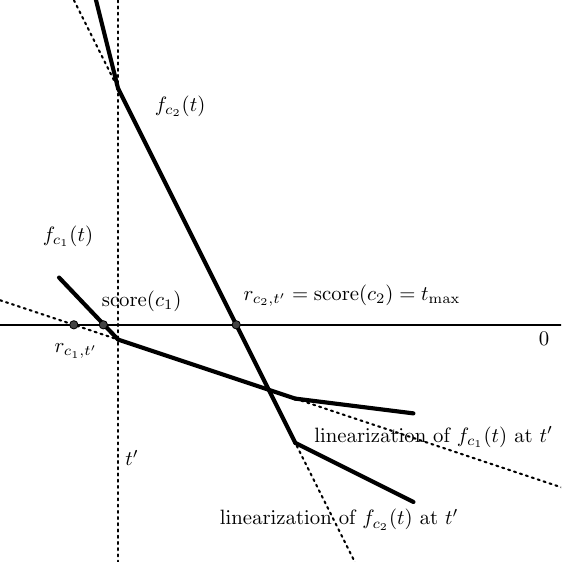}
  \caption{For each candidate $c_i$, the root $\score(c_i)$ of function $f_{c_i}(t)$ lies to the right of $r_{c_i, t'}$, the root of its linearization at $t'$. These two roots coincide for $c_2=c_{\max}$. }
  \label{fig:maxscore}
\end{figure}

On the other hand, for the candidate $c_{\max}$ with highest score $t_{\max}$, and for $x=t'$, the output of $\interval(A,w)$, we have that the corresponding linearization coincides with function $f_{c_{\max}}(t)$ in the interval $[t', t_{\max}]$, so the linearization root $r_{c_{\max}, t'}$ equals the function root $t_{\max}=\score(c_{\max})$. 
See Figure~\ref{fig:maxscore}. Consequently, %
\begin{align*}
r_{c_{\max}, t'} &= \score(c_{\max}) = \max_{c'\in C\setminus A} \score(c') 
	\geq \max_{c'\in C\setminus A} r_{c', t'} \geq r_{c_{\max}, t},
\end{align*}
i.e., $c_{\max}$ is simultaneously the candidate with highest score and the one with highest linearization root at $t'$, and these values coincide. 
We use this fact to find the candidate and its score. We formalize these observations in Algorithm~\ref{alg:maxscore} and the lemma below.

\begin{algorithm}[htb]
\SetAlgoLined
\KwData{Partial solution $(A,w)$.}
Let $t'\leftarrow \interval(A,w)$\;

\For{each voter $n\in N$}{
Compute $p_n:=s_n-\sum_{c\in A\cap C_n: \ \supp_w(c)< t'} w_{nc}$\;
Compute $q_n:=\sum_{c\in A\cap C_n: \ \supp_w(c)\geq t'} w_{nc}/\supp_w(c)$\;
}
\lFor{each candidate $c'\in C\setminus A$}{
compute $r_{c', t'}=\frac{\sum_{n\in N_{c'}} p_n}{1+\sum_{n\in N_{c'}} q_n}$}
Find a candidate $c_{\max}\in\argmax_{c'\in C\setminus A} r_{c', t'}$\;
\Return $(c_{\max}, r_{c_{\max}, t'})$\;
 \caption{$\maxscore(A,w)$}
\label{alg:maxscore}
\end{algorithm}

\begin{lemma}\label{lem:maxscore}
For a partial solution $(A,w)$, $\maxscore(A,w)$ runs in time $O(|E|\cdot \log k)$ and returns a tuple $(c_{\max}, t_{\max})$ such that $c_{\max}\in C\setminus A$ and $t_{\max}=\score(c_{\max})=\max_{c'\in C\setminus A} \score(c')$.
\end{lemma}
\begin{proof}
The correctness of the algorithm follows from the arguments above. 
Each of the \textbf{for} loops executes in time $O(|E|)$ because in each one of them each edge is examined at most once. 
The running time is dominated by the call to algorithm $\interval(A,w)$, taking time $O(|E|\cdot \log k)$.
\end{proof}

This completes the proof of Theorem~\ref{thm:runtimes}. 
We highlight again that the heuristic for candidate selection in $\phragmms$ runs in time $O(|E|\cdot \log k)$, thus almost matching the complexity of the heuristic in $\phragmen$ which is $O(|E|)$ per iteration. 

Next, we reconsider the complexity of $\phragmms$ (Algorithm~\ref{alg:balanced}). 
At the start of each iteration with current partial solution $(A,w)$, notice by Lemma~\ref{lem:315localoptimality} that the highest score $t_{\max}$ must be lower than the least member support $t_1=supp_w(A)$. So, $t_{\max}$ lies in the interval $[0,t_1]$, and we can skip the computation of Algorithm $\interval(A,w)$ as we know that it would return $t'=0$. 
Without this computation, $\maxscore(A,w)$ (Algorithm~\ref{alg:maxscore}) runs in time $O(|E|)$, so the runtime of a full iteration of $\phragmms$ can be performed in time $O(\bal + |E|)$, down from $O(\bal + |E|\cdot \log k)$ as was established in Section~\ref{s:315}.

Finally, we discuss some similarities and differences between the $\phragmms$ and $\phragmen$ heuristics. 
For the sake of completeness, we present here the $\phragmen$ algorithm explicitly. 
We note that the version of $\phragmen$ proposed in~\cite{brill2017phragmen} only considers unit votes. 
In Algorithm~\ref{alg:phragmen} we give a generalization that admits arbitrary vote strengths. 
Clearly, each one of the $k$ iterations of the main loop runs in time $O(|E|)$, because each of the two internal \textbf{for} loops examines each edge in $E$ at most once. 

\begin{algorithm}[htb]
\SetAlgoLined
\KwData{Bipartite approval graph $G=(N\cup C, E)$, vector $s$ of vote strengths, target committee size $k$.}
Initialize $A=\emptyset$, $load(n)=0$ for each $n\in N$, and $load(c')=0$ for each $c'\in C$\;
\For{$i=1,2,\cdots k$}{
\lFor{each candidate $c'\in C\setminus A$}{
update $load(c') \leftarrow \frac{1+\sum_{n\in N_{c'}} s_n\cdot load(n)}{\sum_{n\in N_{c'}} s_n}$}
Find $c_{\min}\in \arg\min_{c'\in C\setminus A} load(c')$\;
Update $A\leftarrow A+c_{\min}$\;
\For{each voter $n\in N_{c_{\min}}$}{
Update $load(n)\leftarrow load(c_{\min})$\;
}
}
\Return $A$\;
\caption{$\phragmen$, proposed in~\cite{brill2017phragmen}}
\label{alg:phragmen}
\end{algorithm}

Assume that we consider inserting a candidate $c'\in C\setminus A$ to the partial solution $(A,w)$, and recall that for a voter $n\in N_{c'}$ approving of that candidate, and a threshold $t$, we define 
$$\slack(n,t)=s_n - \sum_{c\in A\cap C_n} w_{nc} \cdot\min \{1, t/\supp_w(c)\}.$$ 
This formula expresses the fact that to each current member $c\in A\cap C_n$, we reduce its edge weight $w_{nc}$ by multiplying it by a factor $\min \{1, t/\supp_w(c)\}$, and use the now-available vote strength from voter $n$ (its \emph{slack}) to give support to the new member $c'$. This edge multiplication factor is somewhat involved but sensible, as it removes a higher fraction of vote from members with higher support, and leaves members with low support untouched; see Section~\ref{s:inserting} for further intuition.

In contrast, in the same context, we claim that the $\phragmen$ heuristic can be thought of as using a constant edge multiplication factor $t/\supp_w(A\cap C_n)$, where we recall that $\supp_w(A\cap C_n):=\min_{c\in A\cap C_n} \supp_w(c)$. This is, of course, a much simpler approach, corresponding to a coarser solution rebalancing method. 

We now prove our claim. Suppose we use the edge multiplication factor above, and consequently define the voter's slack as
\begin{align}\label{eq:alt-slack}
\slack'(n,t):=s_n - \frac{t}{\supp_w(A\cap C_n)}\sum_{c\in A\cap C_n} w_{nc}.
\end{align}
If we define parameterized scores and scores as before, we can retrieve the new score value for candidate $c'$ by finding the root of the function in equation~\eqref{eq:scorefunction}, which is now linear. 
With a similar computation as the one we did for equation~\eqref{eq:linearized}, we obtain
\begin{align*}
\score'(c') 
&=\frac{\sum_{n\in N_{c'}} s_n}{1+\sum_{n\in N_{c'}} \frac{1}{\supp_w(A\cap C_n)} \sum_{c\in A\cap C_n} w_{nc} } \\
&\geq  \frac{\sum_{n\in N_{c'}} s_n}{1+\sum_{n\in N_{c'}} \frac{s_n}{\supp_w(A\cap C_n)}}, 
\end{align*}
where the inequality follows by feasibility (inequality~\ref{eq:feasible}), and is tight if we assume that the current partial solution $(A,w)$ uses up all the vote strength of voter $n$ whenever $A\cap C_n$ is non-empty. If $A\cap C_n=\emptyset$, then $\supp_w(A\cap C_n)=\infty$ by convention and the corresponding term vanishes in the denominator. 
This new score, to be maximized among all unelected candidates, corresponds precisely to the inverse of the \emph{candidate load} being minimized in the $\phragmen$ heuristic; see Algorithm~\ref{alg:phragmen}. The corresponding \emph{voter load} is in turn set to the inverse of $\supp_w(A\cap C_n)$, which the algorithm updates with the assumption that the new candidate $c'$ always becomes the member with least support. This completes the proof of the claim.

In view of this last result, we can say that our new heuristic provides two main advantages with respect to $\phragmen$: 
First, by using edge weights explicitly, the algorithm handles a more robust notion of loads. 
This enables $\phragmms$ to deal with arbitrary input solutions, a fact that we exploit in Appendix~\ref{s:LS}, and in contrast to $\phragmen$ which needs to make assumptions on the structure of the current solution at the beginning of each iteration.  
Second, our heuristic uses a better rebalancing method that provides more slack to voter $v$ for the same threshold $t$. Indeed, identity $\eqref{eq:slack}$ is at least as large as identity $\eqref{eq:alt-slack}$, and usually larger. Hence, new candidates are granted higher scores and are added to the committee with higher supports.

%% file: lazymms.tex
\section{A lazy greedy algorithm}\label{s:lazymms}

In this section we prove Theorem~\ref{thm:2eps} and present $\lazy$, a variant of $\MMS$ (Algorithm~\ref{alg:mms}) that is faster by a factor $\Theta(k)$ and offers virtually the same approximation guarantee.

\begin{algorithm}[htb]
\SetAlgoLined
\KwData{Approval graph $G=(N\cup C, E)$, vector $s$ of vote strengths, committee size $k$, threshold support $t\geq 0$.}
Initialize $A=\emptyset$, $w=0\in\R^E$, and $U=C$ \;
\tcp{$U$ is set of ``uninspected'' candidates} 
\While{$U\neq \emptyset$}{
	Find $c_{\max}\in \argmax_{c'\in U} \score(c')$ \;
	\tcp{try unispected candidate of highest score} 
	Remove $c_{\max}$ from $U$\;
	Compute a balanced weight vector $w'$ for $A+c_{\max}$\;
	\If{$\supp_{w'}(A+c_{\max})\geq t$}{
		Update $A\leftarrow A+c_{\max}$ and $w\leftarrow w'$ \;
		\lIf{$|A|=k$} { \Return $(A,w)$}
	}
}
\Return a failure message\;
\caption{$\lazy$}
\label{alg:lazy}
\end{algorithm}

$\lazy$ (Algorithm~\ref{alg:lazy}) is lazier than $\MMS$ in the sense that for each candidate it inspects, it decides on the spot whether to add it to the current partial solution, if its insertion keeps the least member support above a certain threshold $t$, or permanently reject it. In particular, each candidate entails the computation of a single balanced weight vector, as opposed to $O(k)$ vectors in $\MMS$. 
For a threshold $t\geq 0$ given as input, the algorithm either succeeds and returns a full solution $(A,w)$ with $\supp_w(A)\geq t$, or returns a failure message. 
The idea is then to run trials of $\lazy$ over several input thresholds $t$, performing binary search to converge to a value of $t$ where it flips from failure to success, and return the output of the last successful trial. 
In terms of runtime, our binary search requires only $O(\log (1/\eps))$ trials -- as we shall prove -- and in each trial each of the $O(|C|)$ iterations computes a balanced weight vector in time $\bal$, for an overall complexity of $O(\bal\cdot |C| \log (1/\eps))$. 
In each iteration, the highest score in $U$ can be found in time $O(|E|\log k)$ with a variant of Algorithm $\maxscore$, hence this complexity is dominated by that of computing a balanced vector.  

We start by proving that for low values of $t$, the algorithm is guaranteed to succeed. 
We highlight that in the following proof the order in which the candidate set $C$ is traversed is irrelevant. 

\begin{lemma}\label{lem:success}
If $(A^*, w^*)$ is an optimal solution to the given instance of maximin support, and $t^*=\supp_{w^*}(A^*)$, then for any input threshold $t$ with $0\leq t\leq t^*/2$, $\lazy$ is guaranteed to succeed.
\end{lemma}

\begin{proof}
Assume by contradiction that for some input threshold $t\leq t^*/2$, $\lazy$ fails. 
Thus, after traversing the whole candidate set $C$, the algorithm ends up with a partial solution $(A,w)$ with $|A|<k$ and $\supp_w(A)\geq t$. 
By Lemma~\ref{lem:2sols}, there must be a candidate $c'\in A^*\setminus A$ and a feasible solution $(A+c', w')$ such that $\supp_{w'}(A+c')\geq t$. 
Notice as well that for any subset $S$ of $A+c'$, vector $w'$ still provides it with a support of at least $t$, and consequently any balanced weight vector for $S$ also provides it with a support of at least $t$. 
This implies that at whichever point the algorithm inspected candidate $c'$, it should have included it in the then-current partial solution, which was a subset of $A$. Hence, $c'$ should be contained in $A$, and we reach a contradiction.
\end{proof}

Next, we establish the number of trials needed to achieve a solution within a factor $(2+\eps)$ from optimal for any $\eps>0$. 

\begin{lemma}\label{lem:lazybinary}
For any $\eps>0$, $O(\log(1/\eps))$ trials of $\lazy$ are sufficient to obtain a solution whose maximin support value is within a factor $(2+\eps)$ from optimal.
\end{lemma}

\begin{proof}
First, we need to compute a constant-factor estimate of the optimal objective value $t^*$. 
One way to do that is, of course, to execute $\phragmms$ (Algorithm~\ref{alg:balanced}), which provides an approximation guarantee of $\alpha=3.15$ and runs in time $O(\bal\cdot k)$.%
\footnote{In fact, it can be checked that the $\phragmms$ algorithm is equivalent to an execution of $\lazy$ with threshold $t=0$.} 
If $t$ is the objective value of its output, and we initialize the variables $t'\leftarrow t/2$, $t''\leftarrow \alpha\cdot t$, then $\lazy$ is guaranteed to succeed for threshold $t'$ and fail for $t''$. We keep these properties as loop invariants as we perform binary search with $\lazy$, in each iteration setting the new threshold value to the geometric mean of $t'$ and $t''$. This way, the ratio $t''/t'$ starts with a constant value $2 \alpha$, and is square-rooted in each iteration. 
By Lemma~\ref{lem:success}, to achieve a $(2+\eps)$-factor guarantee it suffices to find threshold values $t'<t''$ such that $\lazy$ succeeds for $t'$ and fails for $t''$ and whose ratio is bounded by $t''/t'\leq 1+\eps/2$, and return the output for $t'$. If it takes $T+1$ iterations for our binary search to bring this ratio below $(1+\eps/2)$, then $(2\alpha)^{1/2^T} > (1+\eps/2)$, so $T= O(\log (\eps^{-1} \log(\alpha))) = O(\log(\eps^{-1}))$. This completes the proof. 
\end{proof}

Finally, we prove that whenever Algorithm $\lazy$ succeeds and returns a full solution, this solution satisfies PJR. For this, we exploit the order in which we inspected the candidates. This completes the proof of Theorem~\ref{thm:2eps}. 

\begin{lemma}
For any input threshold $t$, at the end of each iteration of $\lazy$ we have that if $(A,w)$ is the current partial solution, then $\supp_w(A)\geq \max_{c'\in C\setminus A} \score_{(A,w)}(c')$. 
Therefore, if the algorithm succeeds and returns a full solution, this solution satisfies PJR.
\end{lemma}

\begin{proof}
The second statement immediately follows from the first one together with Lemma~\ref{lem:localopt}, hence we focus on proving the first statement. 
Fix an input threshold $t$ and some iteration of $\lazy$, and let $(A,w)$ be the partial solution at the end of it. 
We consider three cases. Case 1: If all candidates inspected so far have been added to the solution and not rejected, then up to this point the construction coincides with Algorithm $\phragmms$, and the claim follows by Lemma~\ref{lem:315localoptimality}. 
Case 2: Suppose the last iteration was the first to reject a candidate, and let $c'$ be this candidate. Then, $c'$ has the highest score in $C\setminus A$, and we claim that this score must be below threshold $t$, and hence below $\supp_w(A)$. Otherwise, by Lemma~\ref{lem:insert} we have that Algorithm $\ins(A,w,c',t)$ could find a weight vector that gives $A+c'$ a support above $t$, so a balanced weight vector would also give $A+c'$ a support above $t$ which contradicts the fact that $c'$ was rejected. 
Case 3: If a candidate was rejected in a previous iteration, then at the time of the first rejection we had that the highest score in $C\setminus A$ was below $t$, and this inequality must continue to hold true in further iterations by Lemma~\ref{lem:2balanced}, because scores can only decrease. This completes the proof.
\end{proof}

%% file: proofs.tex
\section{Delayed proofs}\label{s:proofs}

In this section we make frequent references to the definitions and results related to network flow that we present in Appendix~\ref{s:flow}. We start with the following auxiliary result which will be used by some of the proofs in this section.

\begin{lemma}\label{lem:path}
Let $w\in \R^E$ be a feasible weight vector for a given instance, let $c,c'\in C$ be two candidates with $\supp_w(c)<\supp_w(c')$, and suppose there is a simple path $p\in\mathbb{R}^E$ that carries non-zero flow from $c'$ to $c$. If $w+p$ is non-negative and feasible, then $w$ is not balanced for any committee $A$ that contains $c$.
\end{lemma}

\begin{proof}
Fix a committee $A\subseteq C$ that contains $c$. If $c'$ is not in $A$, then $w+p$ provides a greater sum of member supports over $A$ than $w$, so the latter is not balanced as it does not maximize this sum. Now suppose both $c$ and $c'$ are in $A$. Let $\lambda>0$ be the flow value carried by $p$, let $\eps:=\min\{\lambda, (\supp_w(c') - \supp_w(c))/2\}>0$, and let $p'$ be the scalar multiple of $p$ whose flow value is $\eps$. By an application of Lemma~\ref{lem:subflow} over $w$ and $w':=w+p$, and the fact that $p'$ is a sub-flow of $p=w'-w$, we have that $w+p'$ is non-negative and feasible. 
Moreover, vectors $w$ and $w+p'$ clearly provide the same sum of member supports over $A$. Finally, if we compare their sums of member supports squared, we have that
\begin{align*}
\sum_{d\in A}  \supp_w^2(d)-\sum_{d\in A} \supp^2_{w+p'}(d) 
=& \supp_w^2(c) +\supp_w^2(c')-\supp^2_{w+p'}(c) -\supp^2_{w+p'}(c')\\
=& \supp_w^2(c) +\supp_w^2(c') - (\supp_w(c) + \eps)^2 - (\supp_w(c') - \eps)^2 \\
=& 2\eps\cdot (\supp_w(c') - \supp_w(c) - \eps) \\
\geq & 2\eps\cdot(2\eps - \eps)=2\eps^2>0.
\end{align*}

Therefore, $w$ is not balanced for $A$, as it does not minimize the sum of member supports squared.  
\end{proof}

\begin{proof}[Proof of Lemma~\ref{lem:balanced}]
Fix a balanced partial solution $(A,w)$. The first statement is that for any $1\leq r\leq |A|$, function 
$$F_r(w'):=\min_{A'\subseteq A, \ |A'|=r} \sum_{c\in A'} \supp_{w'}(c)$$
is maximized by vector $w$ over all feasible vectors $w'\in\R^E$. 
Assume by contradiction that there is a parameter $r$ and a feasible $w'$ such that $F_r(w')>F_r(w)$. 
We also assume without loss of generality that the members of $A=\{c_1, \cdots, c_{|A|}\}$ are enumerated in such a way that whenever $i<j$, we have $\supp_w(c_i)\leq \supp_w(c_j)$, and if this inequality is tight then $\supp_{w'}(c_i)\leq \supp_{w'}(c_j)$.
With this enumeration we obtain the identity $F_r(w)=\sum_{i=1}^r \supp_w(c_i)$. 
Thus, by our assumption by contradiction, 
$$ \sum_{i=1}^r \supp_{w}(c_i) = F_r(w) < F_r(w') \leq  \sum_{i=1}^r \supp_{w'}(c_i). $$
%
%\begin{align*}
%    \sum_{i=r+1}^{|A|} \supp_{w'}(c_i) &= \sum_{i=1}^{|A|} \supp_{w'}(c_i) - \sum_{i=1}^{r} \supp_{w'}(c_i) \\
%    & = \sum_{i=1}^{|A|} \supp_{w}(c_i) - \sum_{i=1}^{r} \supp_{w'}(c_i) \\
%    & < \sum_{i=1}^{|A|} \supp_{w}(c_i) - \sum_{i=1}^{r} \supp_{w}(c_i) \\
%    & = \sum_{i=r+1}^{|A|} \supp_{w}(c_i). \\
%\end{align*}
%
Consider the flow $f:=w'-w\in\mathbb{R}^E$. 
By the last inequality, set $A_r:=\{c_1, \cdots, c_r\}\subseteq A$ has a negative net excess, so by Theorem~\ref{thm:decomposition} $f$ must have a sub-flow $p$ that is a simple path starting in a vertex outside $A_r$ with positive excess and ending in a vertex $c_i$ in $A_r$ with negative excess, for some $1\leq i\leq r$. 
Now, path $p$ must also start inside committee $A$, as otherwise vector $w+p$ is feasible by Lemma~\ref{lem:subflow} and provides a larger sum of member supports than $w$, contradicting the fact that the latter is balanced for $A$. 
Hence, $p$ starts in an vertex $c_j$ in $A$ for some $r<j\leq |A|$. 
Moreover, by our choice of member enumeration it must be the case that $\supp_w(c_i)<\supp_w(c_j)$, because if the inequality was tight we would have $\supp_{w'}(c_i)<\supp_{w'}(c_j)$, which implies that $c_i$ has a larger excess than $c_j$, $e_f(c_i)> e_f(c_j)$, contradicting the fact that $c_i$ has negative excess and $c_j$ has positive excess.
Finally, by Lemma~\ref{lem:path}, $w$ is not balanced for $A$ (nor for any committee that contains $c_i$), and we reach a contradiction. 

The second statement follows directly from the fact that $w$ maximizes the sum of member supports, and thus all of the aggregate vote strength of represented voters (i.e., voters in $\cup_{c\in A} N_c$) must be directed to members of $A$. 
We move on to the third statement. 
Assume by contradiction that there is a voter $n\in N$ and two candidates $c, c'\in A\cap C_n$ such that $w_{nc}>0$ and $\supp_w(c)>\supp_w(c')$. 
Let $p\in\mathbb{R}^E$ be the simple path of length two that carries a flow of value $w_{nc}$ from $c$ to $c'$ via $n$, i.e, $p_{nc'}=-p_{nc}=w_{nc}$, and $p$ is zero elsewhere. 
It can be checked that $w+p$ is non-negative and feasible, so by Lemma~\ref{lem:path} $w$ is not balanced for $A$, which is a contradiction. 

Finally, we prove that if a feasible weight vector satisfies properties 2 and 3, then it is necessarily balanced for $A$. 
In fact, we claim that all vectors satisfying these properties provide exactly the same list of member supports $(\supp_w(c))_{c\in A}$, and hence all are balanced for $A$. 
Let $w, w'\in\R^E$ be two such weight vectors. It easily follows from feasibility (inequality~\ref{eq:feasible}) and property 2 that both provide the same sum of member supports, namely 
$$\sum_{c\in A} \supp_w(c) = \sum_{c\in A} \supp_{w'}(c) =\sum_{n\in \cup_{c\in A} N_c} s_n.$$
 
Now, assume by contradiction and without loss of generality that there is a candidate $c\in A$ for which $\supp_{w}(c)<\supp_{w'}(c)$, and consider the flow $f:=w'-w$. 
Clearly, all vertices in $N$ and in $C\setminus A$ have zero excess relative to $f$, while $c$ has negative excess. 
By Theorem~\ref{thm:decomposition}, there is a simple path $p$ ending in $c$ and starting in a vertex with positive excess; 
moreover, $p$ is a sub-flow of $f$, which implies that this starting vertex must be in $A$, and in fact all of the candidates visited by $p$ must be in $A$ as well.  

Now, path $p$ alternates between members of $A$ and voters. Hence, there must be three consecutive vertices $c_1, n, c_2$ in it, with $c_1, c_2\in A\cap C_{n}$, such that $c_1$ has a strictly larger excess than $c_2$, i.e., 
$$- \supp_{w'}(c_1)+\supp_w(c_1) > - \supp_{w'}(c_2) + \supp_w(c_2),$$ 

which in turn implies that either $\supp_{w'}(c_1)<\supp_{w'}(c_2)$ holds or $\supp_{w}(c_1)>\supp_w(c_2)$ holds, or both. 
If $\supp_{w'}(c_1)<\supp_{w'}(c_2)$, we reach a contradiction with the fact that $w'$ satisfies property 3 and that $w'_{nc_2}$ must be strictly positive since $p$ (and thus also $f=w'-w$) is directed from $n$ to $c_2$. 
Similarly, if $\supp_w(c_2)<\supp_{w}(c_1)$ we reach a contradiction with the fact that $w$ satisfies condition 3 and that $w_{nc_1}$ must be strictly positive since $p$ is directed from $c_1$ to $n$. 

Properties 2 and 3 can clearly be tested in time $O(|E|)$. This completes the proof of the lemma.
\end{proof}

\begin{proof}[Proof of Lemma~\ref{lem:badexamples}]
In the example of Figure~\ref{fig:example}, the optimum value for  maximin support is clearly 1, achieved for instance by choosing the $k$ honest candidates. 
Hence, if an $\alpha$-approximation algorithm elects $j$ adversarial candidates, one of these candidates must be given a support that is simultaneously at most $1/j$ and at least $1/\alpha$, so $j\leq \alpha$. This proves the first claim.

We continue with the PAV rule. In this example, it can be checked that PAV yields the same result as sequential-PAV, and that honest candidates are elected in order, i.e., $c_1$, then $c_2$, and so on. 
Now, if at some point the rule has elected $i$ honest and $j-1$ adversarial members, with $i+j=k$, the score of the next honest candidate is $(k-i)/(i+1)$, and that of the next adversarial candidate is $1/j$. 
As we always pick the candidate with highest score, the last candidate will be adversarial -- and hence there will be $j$ adversarial candidates elected -- if $1/j > (k-i)/(i+1)=j/(k-j+1)$. 
It can be checked that $j=\sqrt{k}-1/2$ satisfies this inequality, so the rule elects at least this many adversarial representatives. 

We analyze $\phragmen$ next. We take the continuous formulation where every voter starts with zero vote strength and gains strength at a constant speed of one unit per second, candidates have unit cost, and a candidate is elected as soon as its supporters can afford it, spending the corresponding vote strength; see~\cite{lackner2020approval}.  
It will take $1/k+1/(k-1)+\cdots + 1/(k-i)= H_k - H_{k-i-1}$ seconds for the rule to elect $i+1$ honest candidates, where $H_i=\sum_{t=1}^i 1/t$ is the $i$-th harmonic number, and $j$ seconds to elect $j$ adversarial candidates, where honest and adversarial candidates are elected with independent time frames. 
If at some point there are $i$ honest and $j-1$ adversarial candidates with $i+j=k$, the last elected candidate will be adversarial -- and thus there will be $j$ adversarial candidates elected -- if $j< H_k - H_{k-i-1} = H_k - H_{j-1}=\ln(\frac{k}{j}) -o(1)$. 
From this it follows that the rule elects at least $(1-o(1)) \ln k$ adversarial candidates.

We finally consider Rule X, which consists of two phases. 
In the first phase, $k$ units of vote strength are evenly distributed among the $k+1$ voters, i.e., $k/(k+1)$ units per voter, and candidates have a unit cost as before. 
The vote distribution from voters to candidates is somewhat involved, but it suffices to notice that the adversary cannot afford any candidates, honest candidates are elected in order, and the cost of each new candidate is evenly shared among its supporters. 
At the beginning of the election of the $i$-th honest candidate, each of its $k-i+1$ supporters has a vote strength of $k/(k+1) - 1/k - \cdots - 1/(k-i+1)=1-H_{k+1}-H_{k-i}$, so the candidate can be afforded if and only if $(k-i+1)(1-H_{k+1} - H_{k-i})\geq 1$. As a result, there will be $k[1- e^{-1-o(1)}]$ honest candidates elected in the first phase. 
Now, the rule is not specific about how to elect the remaining candidates in the second phase, so it could elect up to $k/e^{1+o(1)}=\Omega(k)$ adversarial candidates. 
If the remaining seats are filled by running $\phragmen$, as suggested by the authors of Rule X~\cite{peters2019proportionality}, then the rule selects at least as many adversarial candidates as $\phragmen$ does, since the adversarial voter still has all of its budget available. This completes the proof.
\end{proof}

\begin{proof}[Proof of Lemma~\ref{lem:2balanced}]
The second statement comparing scores follows directly from the first one and the definitions of slack, parameterized score, and score. Hence we focus on the first statement, i.e., that $\supp_w(c)\geq \supp_{w'}(c)$ for each member $c\in A$.

Consider the flow $f:=w'-w\in\mathbb{R}^E$: it suffices to prove that no member of $A$ has negative excess relative to it. Assume by contradiction that there is such a member $c\in A$ with negative excess.  
By Theorem~\ref{thm:decomposition}, $f$ must have a sub-flow that is a simple path ending in $c$ and starting in vertex with positive excess. 
This starting vertex must be a candidate $c'$ inside $A$, as otherwise vector $w+p$ is feasible by Lemma~\ref{lem:subflow} and offers a greater sum of member supports over $A$ than $w$, which contradicts the fact that $w$ is balanced for $A$. 
Now, the fact that $e_f(c) < 0 < e_f(c')$ implies that 
$$- \supp_{w'}(c) + \supp_{w}(c) <0< - \supp_{w'}(c') + \supp_{w}(c'),$$
which implies that either $\supp_{w}(c)<\supp_{w}(c')$, or $\supp_{w'}(c')<\supp_{w'}(c)$, or both. 
If $\supp_{w}(c)<\supp_{w}(c')$, then by Lemma~\ref{lem:path}, $w$ is not balanced for $A$ (nor for any committee containing $c$). 
Similarly, if $\supp_{w'}(c')<\supp_{w'}(c)$, notice that $w'-p$ is non-negative and feasible by Lemma~\ref{lem:subflow}, so again Lemma~\ref{lem:path} applied to vector $w'$ and path $-p$ (which starts in $c$ and ends in $c'$) implies that $w'$ is not balanced for $A'$ (nor for any committee containing $c'$). 
In either case we reach a contradiction.
\end{proof}

\begin{proof}[Proof of Lemma~\ref{lem:Lebesgue}]
Recall that for any set $A\subseteq \mathbb{R}$, the indicator function $1_A:\mathbb{R}\rightarrow \mathbb{R}$ is defined as $1_A(t)=1$ if $t\in A$, and $0$ otherwise. For any $i\in I$, we can write
$$ f(x_i) = \int_{0}^{f(x_i)} dt = \int_0^{\lim_{x\rightarrow \infty} f(x)} 1_{(-\infty, f(x_i)]}(t)dt,$$
and thus
\begin{align*}
    \sum_{i\in I} \alpha_i f(x_i) &= \int_0^{\lim_{x\rightarrow \infty} f(x)} \Big(\sum_{i\in I} \alpha_i 1_{(-\infty, f(x_i)]}(t)\Big)dt \\
		&= \int_0^{\lim_{x\rightarrow \infty} f(x)} \Big(\sum_{i\in I: \ f(x_i)\geq t} \alpha_i \Big)dt.
\end{align*}
This is a Lebesgue integral over the measure with weights $\alpha_i$. Now, conditions on function $f(x)$ are sufficient for its inverse $f^{-1}(t)$ to exist, with $f^{-1}(0)=\chi$. Substituting with the new variable $x=f^{-1}(t)$ on the formula above, where $t=f(x)$ and $dt=f'(x)dx$, we finally obtain
$$\sum_{i\in I} \alpha_i f(x_i) =\int_{\chi}^{\infty} \Big( \sum_{i\in I: \ x_i\geq x} \alpha_i \Big)\cdot f'(x)dx,$$
as claimed.
\end{proof}

\begin{proof}[Proof of Lemma~\ref{lem:subflow}]
We prove the claim only for $w+f'$, as the proof for $w'-f'$ is symmetric. 
For each edge $nc\in E$, the value of $(w+f')_{nc}$ must fall between $w_{nc}$ and $(w+f)_{nc}=w'_{nc}$. As both of these values are non-negative, the same holds for $(w+f')_{nc}$. 
Notice now from inequality \eqref{eq:feasible} that proving feasibility corresponds to proving that the excess $e_{w+f'}(n)$ is at most $s_n$ for each voter $n\in N$. We have 
$$e_{w+f'}(n) = \sum_{c\in C_n} (w+f')_{nc}= \sum_{c\in C_n} ( w_{nc} + f_{nc}') = e_w(n) + e_{f'}(n). $$
If the excess $e_{f'}(n)$ is non-positive, then $e_{w+f'}(n)\leq e_w(n) \leq s_n$, where the last inequality holds because $w$ is feasible. 
Otherwise, we have  $0< e_{f'}(n)\leq e_{f}(n)$, and thus $e_{w+f'}(n)\leq e_w(n) + e_{f}(n) = e_{w+f}(n) = e_{w'}(n) \leq s_n$, since $w'$ is feasible. This completes the proof.
\end{proof}